\documentclass[12pt,a4paper]{article}

\usepackage{amscd,amsfonts,amsmath,amssymb,amsthm,bbm,bm,latexsym,mathrsfs}
\usepackage{epsfig,graphics,graphicx,natbib,subfigure,longtable}
\usepackage[english]{babel}
\usepackage[usenames]{color}
\usepackage{bookmark}
\usepackage{booktabs}
\usepackage{sgame}
\usepackage[top=1.7in, bottom=1.7in, left=0.5in, right=0.5in]{geometry}
\allowdisplaybreaks
\makeatother

\newtheorem{theorem}{Theorem}[section]

\newtheorem{definition}{Definition}[section]

\theoremstyle{remark}

\theoremstyle{plain}

\newcommand{\real}{\ensuremath{\mathbb{R}}}

\newcommand{\one}{\mathbbm{1}}

\begin{document}

\title{\vspace{-50pt}{Loss-based approach to two-piece location-scale distributions with applications to dependent data}} 

\author{
\hspace{-20pt} 
Fabrizio Leisen\textsuperscript{a} \hspace{15pt}
 Luca Rossini\textsuperscript{b} \hspace{15pt} Cristiano Villa\textsuperscript{a}
        \\ 
        \vspace{5pt}
        \\
        {\centering {\small \textsuperscript{a}University of Kent, U.K. \hspace{5pt}
        \textsuperscript{b}Vrije Universiteit Amsterdam, The Netherlands}}}

\date{}
\maketitle

\abstract{
Two-piece location-scale models are used for modeling data presenting departures from symmetry. In this paper, we propose an objective Bayesian methodology for the tail parameter of two particular distributions of the above family: the skewed exponential power distribution and the skewed generalised logistic distribution. We apply the proposed objective approach to time series models and linear regression models where the error terms follow the distributions object of study. The performance of the proposed approach is illustrated through simulation experiments and real data analysis. The methodology yields improvements in density forecasts, as shown by the analysis we carry out on the electricity prices in Nordpool markets. \\

\textbf{Keywords: } Bayesian inference, loss-Based prior, objective Bayes,,  electricity prices.\\

\textbf{MSC}: 62F15, 62M10.

}


\section{Introduction}
\label{Intro}

Two-piece location-scale models have been mainly used for modeling data exhibiting departures from symmetry. Moreover, some specific two-piece location-scale distributions have been employed in finance to represent  the errors in GARCH-type models, see \cite{ZhuZinde09}, \cite{ZhuGal11}. Different mechanisms have been presented to obtain skewed distributions by modifying symmetric distributions \citep{Azzalini85, FerSteel98, MudHud00}. Recently, the objective Bayesian literature focused on this class of models. Firstly, \cite{RubSte14} derived the Jeffreys rule prior and the independence Jeffreys priors for different families of skewed distributions. They show that Jeffreys priors for some distributions, such as the skewed Student-$t$,  lead to improper posterior distributions. Conversely, reference priors have shown to be more suitable for the above class of distributions, see  \cite{TuWangSun16}. 

In this work, we introduce a novel objective prior for some distributions of the class of two-piece location-scale models, such as the skewed exponential power distribution (SEPD) and the skewed generalized logistic distribution (SGLD). Following \cite{LeiMarVilla16}, we introduce a Bayesian approach obtained by applying the loss-based prior discussed in \cite{VillaWalker15}. In particular, we derive the loss-based prior for the parameter that controls heaviness of the tails of the distribution. 

In the literature, the asymmetric Laplace distribution (ALD) or the asymmetric Student-$t$ distribution (AST) have gained importance in a wide range of disciplines, such as economics \citep{ZhaoZhang07, LeiMarVilla16}, financial analysis \citep{ZhuGal10, KozPod01, Harvey2016} and microbiology \citep{RubioSteel11}. 
However, the application of the SEPD and SGLD to represent the errors of time series and regression models, has received limited attention in the context of objective Bayesian analysis. The aim of this paper is to contribute to the above research area by introducing an information theoretical approach to address inference on the tail parameter of the two skewed distributions. 

As currently there is a growing interest in electricity prices (see \cite{Weron2014} and \cite{Nowotarski2018} for a review),  we will contribute to the analysis of monthly electricity prices in the Nordpool market, in particular for Denmark and Finland through an autoregressive model with errors distributed as a SEPD. Compared to the standard frequentist autoregressive approach, which is the benchmark in the literature (see  \cite{Conejo2005}, \cite{Misiorek2006} and \cite{Maciejowska2015}), we can show that our methodology  improves the density forecasting. 
In addition, we consider a linear regression model where the residuals are SGLD with a loss-based prior on the tail parameter. We illustrate the above model by studying the Small Cell Cancer data set in \cite{Ying1995} and in \cite{RubioYu2017}.

The structure of this document is as follows. In Section \ref{Prelim} we introduce the general two-piece location-scale distribution and discuss special distributions further developed in the paper, such as the SEPD and the SGLD. Section \ref{Prior} focuses on the derivation of the objective priors for the parameters of the models here considered. In Section \ref{Simu} we analyse the frequentist properties of the proposed prior using data simulated from regression models and time series models. Section \ref{Real} deals with real data, in particular we model electricity prices and a cancer dataset. Final discussion points and conclusions are presented in Section \ref{Concl}.

\section{Two-piece location-scale models}
\label{Prelim}

As described in \cite{RubSte14}, in the simple univariate location-scale model it is possible to induce skewness by the use of different scales on both sides of the model and using three different scalar parameters. Firstly, we introduce a general definition of two-piece location-scale models and then we describe different distributions of this family. The general two-piece location-scale density has the following form:
\begin{equation}
g(y|\mu,\sigma_1,\sigma_2,\alpha)=\frac{1}{\sigma_1}f\left( \frac{y-\mu}{2\alpha\sigma_1}\right) \one_{(-\infty,\mu)}(y) + \frac{1}{\sigma_2}f\left( \frac{y-\mu}{2(1-\alpha)\sigma_2}\right) \one_{(\mu,\infty)}(y),
\label{GenFun}
\end{equation}
where $f$ is an absolutely continuous distribution on $\mathbb{R}$, $\mu\in\mathbb{R}$ is the location parameter, $\sigma_1 \in \mathbb{R}^+$ and $\sigma_2 \in \mathbb{R}^+$ are the separate scale parameters and $\alpha \in (0,1)$ is the skewness parameter. In this paper, we follow \cite{RubSte14} and assume $f$ to be symmetric with a single mode at zero, which means that $\mu$ is the mode of the density in \eqref{GenFun}. Hereafter, we assume $\sigma_1=\sigma_2=\sigma$ and we focus on three particular two-piece location-scale models: the skewed Student-$t$ distribution (SST), the skewed exponential power distribution (SEPD) and the skewed generalized logistic distribution (SGLD). These distributions depend on an additional parameter $p$ which controls the behaviour of the tails. The SST is defined as follows. 

\begin{definition}[Skewed Student-$t$ distribution]
Assume $\mu \in \mathbb{R}$ the location parameter, $\sigma >0$  the scale parameter,  $\alpha \in (0,1)$ the skewness parameter and $p>0$  the tail parameter. We define the skewed Student-t distribution as
\begin{equation}
f_{\text{SST}}(y|\alpha,p,\mu,\sigma)=\left\{
\begin{array}{ll}
\dfrac{K(p)}{\sigma} \left[1+\dfrac{1}{p}\left(\dfrac{y-\mu}{2\alpha \sigma}\right)^2\right]^{-\frac{p+1}{2}}, & y \leq \mu,  \\
\dfrac{K(p)}{\sigma} \left[1+\dfrac{1}{p}\left(\dfrac{y-\mu}{2(1-\alpha) \sigma}\right)^2\right]^{-\frac{p+1}{2}}, & y > \mu, \label{SST}
\end{array}
\right.
\end{equation}
where 
 \begin{equation}
 K(p)=\frac{\Gamma[(p+1)/2]}{\sqrt{\pi p}\,\Gamma(p/2)} \notag
 \end{equation}
is a function depending on the tail parameter $p$.
\end{definition}

For a more detailed description of the properties of the SST distribution, see \cite{FerSteel98} and \cite{ZhuGal10}. The SST has some special cases: if $\alpha=1/2$, it is the usual Student-$t$ with $p$ degrees of freedom; if $p=1$, is the skewed Cauchy, while for  $p\rightarrow\infty$, it converges to the skewed normal distribution. \cite{LeiMarVilla16} have proposed a loss-based prior for the tail parameter $p$ of the SST distribution. Therefore, hereafter we will give limited attention to this distribution and we will focus on the remaining two distributions.

The first distribution of interest in our analysis accomodates heavy tails as well as  skewness and is defined as follows. 

\begin{definition}[Skewed exponential power distribution] 
Let us define $\mu \in \mathbb{R}$, the location parameter, $\sigma >0$, the scale parameter, $\alpha \in (0,1)$ the skewness parameter and $p>0$  the tail parameter.  The skewed exponential power distribution has the form
\begin{equation}
f_{\text{SEPD}}(y|\alpha,p,\mu,\sigma)=\left\{
\begin{array}{ll}
\frac{K(p)}{\sigma}\exp{\left\{-\frac{1}{p}\left|\frac{y-\mu}{2\alpha \sigma}\right|^{p}\right\}}, & y \leq \mu,  \\
\frac{K(p)}{\sigma}\exp{\left\{-\frac{1}{p}\left|\frac{y-\mu}{2(1-\alpha) \sigma}\right|^{p}\right\}}, & y > \mu, 
\label{SEPD}
\end{array}
\right.
\end{equation}
with normalizing constant
\begin{equation}
K(p)=\frac{1}{2p^{1/p}\Gamma\left(1+\frac{1}{p}\right)}. \notag
\end{equation}
\end{definition}

The SEPD has been studied in \cite{FerSteel98}, \cite{Kom07} and \cite{ZhuZinde09}. In detail, for $p=1$ the SEPD becomes a skewed Laplace distribution, and for $p=2$ is a skewed normal distribution. For values of $p\to \infty$, we have that the SEPD reduces to an uniform distribution. 

The second distribution we will study, is built on a Beta transformation of the logistic distribution (as described in \cite{Jones04}):
\begin{equation}
f(x)=p[S(x)]s(x), \label{BB}
\end{equation}
where $p(\cdot)$ is the probability density function of a Beta distribution with parameters $(p,p)$ and $S(x)$ and $s(x)$ are, respectively, the cumulative distribution function and the probability density function of the logistic distribution
\begin{equation}
s(x)=\frac{\exp{\left\{-\frac{x-\mu}{\sigma}\right\}}}{\sigma \left(1+\exp{\left\{-\frac{x-\mu}{\sigma}\right\}}\right)^2}, \quad x\in \mathbb{R}. \notag 
\end{equation}
Note that, $f(x)$ in \eqref{BB} is also known as the logistic distribution of the III type. Its skewed version is defined as follows.  

\begin{definition}[Skewed generalized logistic distribution]
Assume $\mu \in \mathbb{R}$, the location parameter, $\sigma >0$, the scale parameter, $\alpha \in (0,1)$ the skewness parameter and $p>0$  the tail parameter. We define the skewed generalized logistic distribution as
\begin{equation}
f_{\text{SGLD}}(y|\alpha,p,\mu,\sigma)=\left\{
\begin{array}{ll}
\frac{1}{\sigma B(p,p)}\frac{\left(\exp{\left\{-\frac{y-\mu}{2\alpha \sigma}\right\}}\right)^p}{\left(1+\exp{\left\{-\frac{y-\mu}{2\alpha \sigma}\right\}}\right)^{2p}}, & y \leq \mu,  \\
\frac{1}{\sigma B(p,p)}\frac{\left(\exp{\left\{-\frac{y-\mu}{2(1-\alpha)\sigma}\right\}}\right)^p}{\left(1+\exp{\left\{-\frac{y-\mu}{2(1-\alpha)\sigma}\right\}}\right)^{2p}}, & y > \mu, 
\label{SLD}
\end{array}
\right.
\end{equation}
where $B(p,p)$ is the beta function. 
\end{definition}

 
\section{The objective prior distribution}
\label{Prior}

Through Bayes theorem, we obtain the posterior, given data $\textbf{x}=(x_1,\dots,x_n)$, by combining the likelihood function and the prior. That is
\begin{equation}
\pi(\alpha,p,\mu,\sigma|\textbf{x}) \propto L(\textbf{x}|\alpha,p,\mu,\sigma) \pi(\alpha,p,\mu,\sigma),
\label{Pos}
\end{equation}
where $L(\textbf{x}|\alpha,p,\mu,\sigma)=\prod_{i=1}^n f(x_i|\alpha,p,\mu,\sigma)$ is the likelihood function, and $\pi(\alpha,p,\mu,\sigma)$ is the prior distributions for all the parameters of the two-piece location-scale. Assuming some degree of independence of prior knowledge about the parameters, the prior distribution can be factorized as 
\begin{equation}
\pi(\alpha,p,\mu,\sigma) \propto \pi(p|\alpha,\mu,\sigma) \pi(\mu,\sigma) \pi(\alpha).
\label{preq}
\end{equation}
In the next section we will show that for  the models under consideration, the prior on the parameter $p$ does not depend on $\alpha$, $\mu$ and $\sigma$. As such, we can write $\pi(p|\alpha,\mu,\sigma)=\pi(p)$.

\subsection{Loss-based prior for $p$}
\label{ObB}
The main focus of this paper is to make  inference on the parameter $p$.
Without loss of generality, $p$ is considered discrete taking values in $\mathbb{N}$.  This is motivated by the fact that seldom the amount of information about $p$ in the data is sufficient to discern between distributions that differ in $p$ less than one. For instance, this is a well known fact for the Student-$t$ distribution. 

\cite{VillaWalker15} introduced a method for specifying an objective prior for discrete parameters. Consider the general two-piece location scale distribution 

\begin{equation}
f_p^{\alpha,\mu,\sigma}(y)=\left\{
\begin{array}{ll}
\frac{1}{\sigma}f_p\left( \frac{y-\mu}{2\alpha\sigma}\right), & y \leq \mu,  \\
\frac{1}{\sigma}f_p\left( \frac{y-\mu}{2(1-\alpha)\sigma}\right), & y > \mu,
\label{GenTwoPiece}
\end{array}
\right.
\end{equation}
which corresponds to the SEPD if $f_p$ is the exponential power distribution and to the SGLD if $f_p$ coincides with the distribution displayed in equation \eqref{BB}. 

The density function $f_p^{\alpha,\mu,\sigma}(y)$ is characterised by the unknown discrete parameter $p$. The idea is to assign a \emph{worth} to each parameter value by objectively measuring what is lost if the value is removed, and it is the true one. The loss is evaluated by applying the well known result in \cite{Berk66} stating that, if a model is misspecified, the posterior distribution asymptotically accumulates on the model which is the nearest to the true one, in terms of the Kullback--Leibler divergence.
Therefore, the \emph{worth} of the parameter value $p$ is represented by the Kullback--Leibler divergence $D_{\text{KL}}\left(f_p^{\alpha,\mu,\sigma}\|f_{p'}^{\alpha,\mu,\sigma}\right)$, where $p^\prime\neq p$ is the parameter value that minimizes the divergence. To link the \emph{worth} of a parameter value to the prior mass, \cite{VillaWalker15} use the self-information loss function. This particular type of loss function measures the loss in information contained in a probability statement \citep{MerFed98}. As we now have, for each value of $p$, the loss in information measured in two different ways, we simply equate them obtaining the loss-based prior:
\begin{equation}
\pi(p) \propto \exp{\biggl\{\min_{p' \ne p} D_{\text{KL}}\left(f_p^{\alpha,\mu,\sigma}\|f_{p'}^{\alpha,\mu,\sigma}\right)\biggr\}} -1,\label{op}
\end{equation}
where 
\begin{equation}
D_{\text{KL}}\left(f_p^{\alpha,\mu,\sigma}\|f_{p'}^{\alpha,\mu,\sigma}\right)=\int f_p^{\alpha,\mu,\sigma}(y) \log{\left\{\frac{f_p^{\alpha,\mu,\sigma}(y)}{f_{p'}^{\alpha,\mu,\sigma}(y)}\right\}} \, dy 
\notag
\end{equation}
is the Kullback--Leibler divergence. 

Following \cite{LeiMarVilla16}, we introduce a theorem (which proof is in Appendix \ref{AppA}) to study the form of the Kullback--Leibler divergence and consequently of the loss-based prior for the tail parameter $p$. 

\begin{theorem} \label{ThKL}
Let $f_p^{\alpha,\mu,\sigma}$ be the density function displayed in equation \eqref{GenTwoPiece} which could be either the SEPD or the SGLD. Then, 
\begin{equation}
D_{\text{KL}}\left(f_p^{\alpha,\mu,\sigma}\|f_{p'}^{\alpha,\mu,\sigma}\right)=D_{\text{KL}}\left(f_{p}^{\alpha=0.5, \mu=0,\sigma=1}\|f_{p'}^{\alpha=0.5, \mu=0,\sigma=1}\right) \notag 
\end{equation} 
for every $p \ge 1$.
\end{theorem}
In other words, Theorem \ref{ThKL} shows that the loss-based prior distribution for the tail parameter $p$ does not depend from the skewness parameter $\alpha$, the location $\mu$ and the scale $\sigma$. Hence, the prior can be written as $\pi(p|\alpha,\mu,\sigma) = \pi(p)$. 

The following theorem derives the  closed form of the Kullback--Leibler divergence for the SEPD. Its proofs, which can be foundin Appendix \ref{AppA}, leverages on the result in Theorem \ref{ThKL}.
\begin{theorem} \label{KLSEPD}
Let $f_p^{\alpha,\mu,\sigma}$ be the SEPD, with skewness parameter $\alpha\in(0,1)$, location parameter $\mu \in \real$, scale parameter $\sigma \in \real_+$ and tail parameter $p \in \{1,2,\dots\}$ as described in equation \eqref{SEPD}. Then, the Kullback--Leibler divergence between two SEPDs that differ in the tail parameter only is:
\begin{equation}
D_{\text{KL}}\left(f_p^{\alpha,\mu,\sigma}\|f_{p'}^{\alpha,\mu,\sigma}\right)= \log{K(p)}-\log{K(p')}-p^{-1}+\frac{p^{\frac{p'}{p}}}{p'}\frac{\Gamma \left(\frac{p'+1}{p}\right)}{\Gamma \left(\frac{1}{p}\right)}.\notag
\end{equation}
\end{theorem}
By applying the result of Theorem \ref{KLSEPD} into equation \eqref{op}, we can then derive the loss-based prior for the SEPD. From Table \ref{TKL30} we see that the minimum Kullback--Leibler divergence is attained for $p^\prime=p+1$ when $p\leq3$ and for $p^\prime=p-1$ for $p>3$. As such, the prior on $p$ is:
$$
\pi(p)\propto
\begin{cases}
\exp\left[\log K(p) - \log K(p+1) - p^{-1} + \dfrac{p^{\frac{p+1}{p}}}{p+1}\dfrac{\Gamma\left(\dfrac{p+2}{p}\right)}{\Gamma\left(\dfrac{1}{p}\right)}\right]-1, & \mbox{for } p\leq3,\\
\exp\left[\log K(p) - \log K(p-1) - p^{-1} + \dfrac{p^{\frac{p-1}{p}}}{p-1}\dfrac{1}{\Gamma\left(\dfrac{1}{p}\right)}\right]-1, & \mbox{for } p>3. \label{lbprior_sepd}
\end{cases}
$$

%
We have numerically verified that the above prior for $p$ is proper for $p=\{1,2,3,\ldots,\infty\}$, therefore yielding a proper posterior.

To derive the loss-based prior for the parameter $p$ of the SGLD, we consider the following Theorem \ref{KLSLD} (which proof is in the Appendix \ref{AppA}), giving the expression of the Kullback--Leibler divergence between two SGLDs.

\begin{theorem} \label{KLSLD}
Let $f_p^{\alpha,\mu,\sigma}$ be the SGLD with skewness parameter $\alpha\in(0,1)$, location parameter $\mu \in \real$, scale parameter $\sigma \in \real_+$ and tail parameter $p \in \{1,2,\dots\}$, respectively as described in equation \ref{SLD}. Then, the Kullback--Leibler divergence between two SGLDs that differ in the tail parameter only is:
\begin{equation}
D_{\text{KL}}\left(f_p^{\alpha,\mu,\sigma}\|f_{p'}^{\alpha,\mu,\sigma}\right)=\log{\left[\frac{B(p',p')}{B(p,p)} \right]} +2(p-p')[\psi(p)-\psi(2p)],\label{KL_Lo}
\end{equation}
where $\psi(p)$ is the digamma function.
\end{theorem}
From Table \ref{TKL_LD30} we see that the Kullback--Leibler divergence between two SGLD is minimised for $p^\prime=p+1$, and thus the loss-based prior is as follows:
\begin{equation}
\pi(p) \propto \frac{p}{2(2p+1)} \exp{\left\{2\left[\psi(2p)-\psi(p)\right]\right\}}-1,\label{OBSLD}
\end{equation}
which is proper, as we have numerically verified.

\subsection{Non-informative prior for the parameters $\alpha$, $\mu$ and $\sigma$.}
\label{PrAMS}
In line with the minimally informative focus of the paper, we have selected objective priors for the other parameters of the considered distributions. That is, we have considered Jeffreys priors for $\alpha$, $\mu$ and $\sigma$. As mentioned at the beginning of Section 3, we assume that the prior information on the true value of the parameters is independent. As such, we can consider, not only $\pi(\alpha)$ on its own, but we can also factorise the prior of the location and the scale parameters; that is $\pi(\mu,\sigma)=\pi(\mu)\pi(\sigma)$. The Jeffreys prior for $\mu$ and $\sigma$ is then proportional to $1/\sigma$, which is obtained by considering the Jeffreys prior for a location parameter, $\pi(\mu)\propto1$, and the Jeffreys prior for a scale parameter, $\pi(\sigma)\propto 1/\sigma$. Both these priors are extensively discussed in \cite{Jeff61}. It is worthwhile to note that the above considerations recover the well-known reference prior for the pair $(\mu,\sigma)$ \citep{BerBerSun09}.

Finally, the Jeffreys prior for the skewness parameter $\alpha$ has been introduced in \cite{RubSte14}, and it shows to be a Beta distribution with both parameters equal to $1/2$. That is $\pi(\alpha)\sim \mbox{Be}(1/2,1/2)$.

\section{Simulation studies}
\label{Simu}
It is important to analyse the performances of objective priors by studying the frequentist properties of the posterior distributions they yield to. As such, the aim of this section is to present simulation studies concerning the objective priors for $p$, as defined in Section 3, for the considered two-piece location-scale models discussed in this work. In particular, we study time series where the residual error terms follow a SEPD and regression models with error terms that follow a SGLD. 

\subsection{SEPD simulation study}

In this simulation exercise, we study an autoregressive (AR) model, where we assume the lag order equal to 1. Therefore, the AR model has the following form:
\begin{equation}
y_t =\phi_1 y_{t-1} + \varepsilon_t, \quad t=1,\dots,T, \label{AR}
\end{equation}
where we assume that the residual errors  $\varepsilon_t$ follow a $\text{SEPD}(\alpha,p,\mu,\sigma)$ with $p=1,\dots, 20$, $\alpha\in\{0.3,0.5,0.8\}$, $\mu = 0$ and $\sigma=1$. The parameter $\phi_1$ is set equal to $0.5$.  Finally, we consider sample sizes of $T=100$ and $T=250$. The analysis has been carried by assuming the loss-based prior on $p$ as defined in Section 3.1. The prior for the remaining three parameters of the SEPD, has been fixed as explained in Section 3.2.  For the parameter $\phi_1$, we assume a Zellner prior \citep{Zellner86} with $g=T$, that is $N(0,T(\sum_{i=1}^{T-1} y_i^2)^{-1})$.
 For each of the above scenarios, we have generated 250 random samples, as described in the Appendix \ref{App_sampling}, and computed the frequentist coverage of the 95\% posterior credible interval for $p$, and the relative square root of the mean squared error $\sqrt{\mbox{MSE}(p)}/p$. ~The coverage measures the frequency of which the true parameter value for $p$ is included in the 95\% credible interval of the posterior distribution of the parameter. Ideally, this value should be close to 0.95. The MSE allows to have a measure of the accuracy of the estimate, intended as the posterior mean for $p$.
\begin{figure}[h!]
\centering
\begin{tabular}{cc}
	{\includegraphics[width=6.75cm]{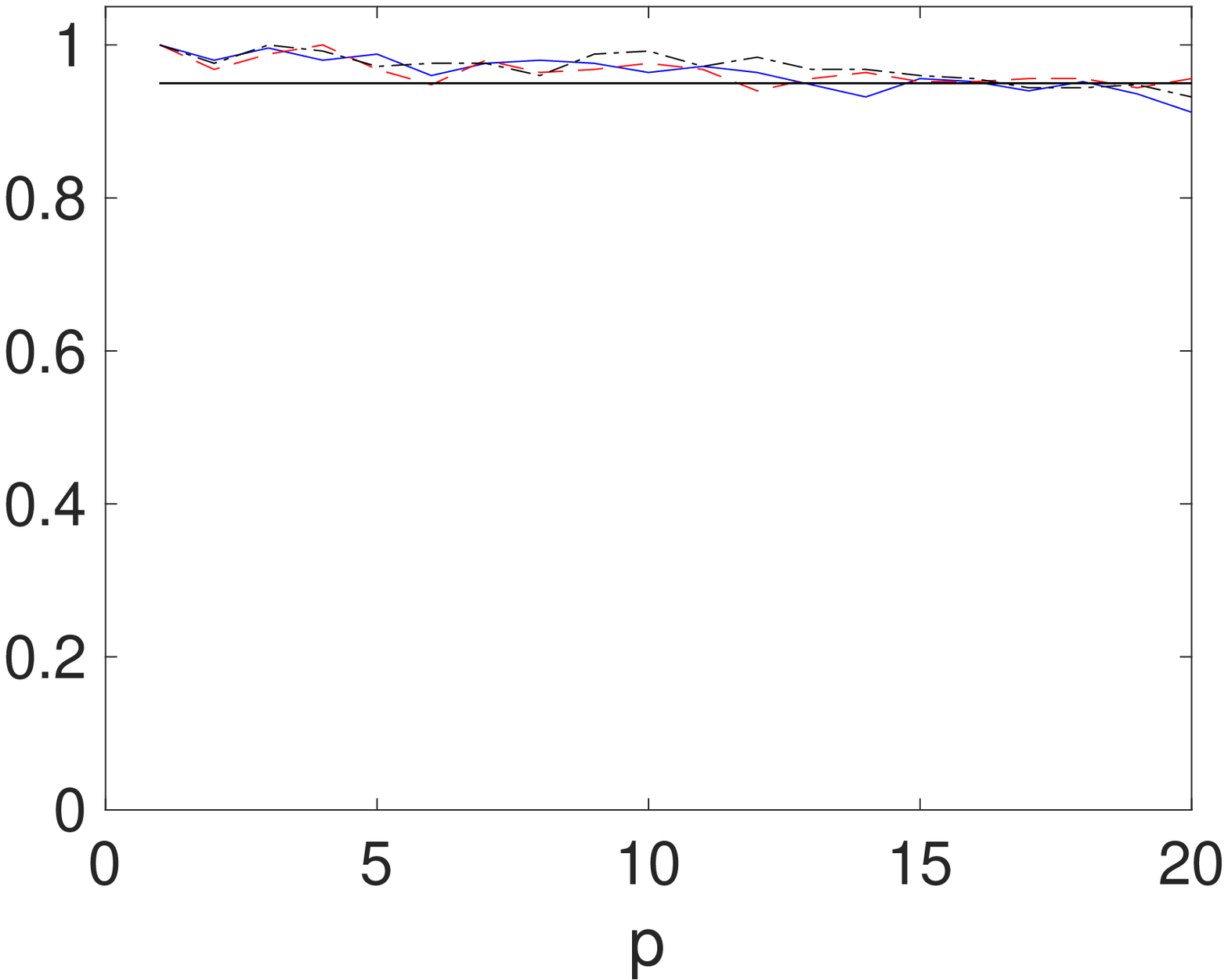}} &
	{\includegraphics[width=6.75cm]{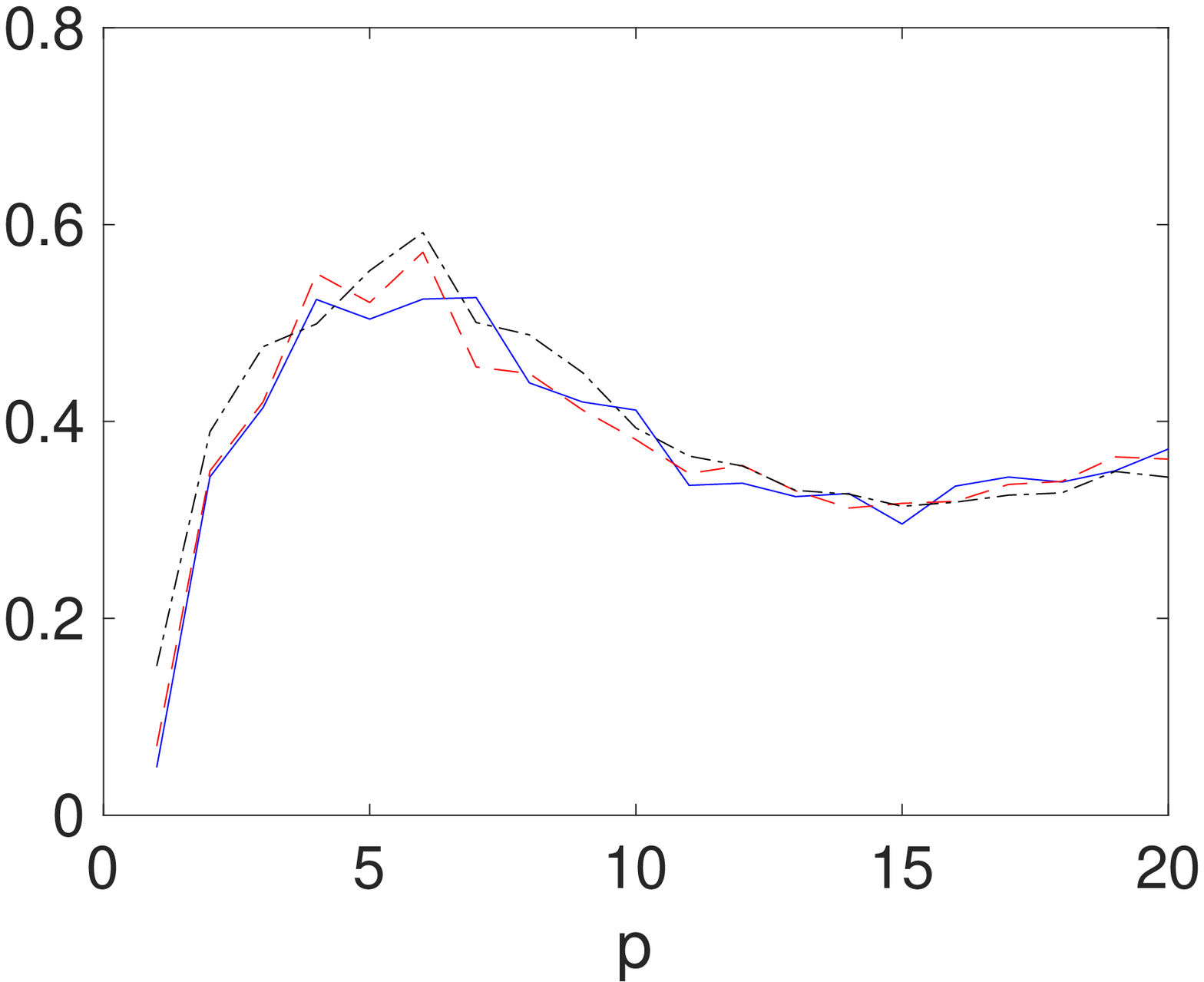}} \\
	{\includegraphics[width=6.75cm]{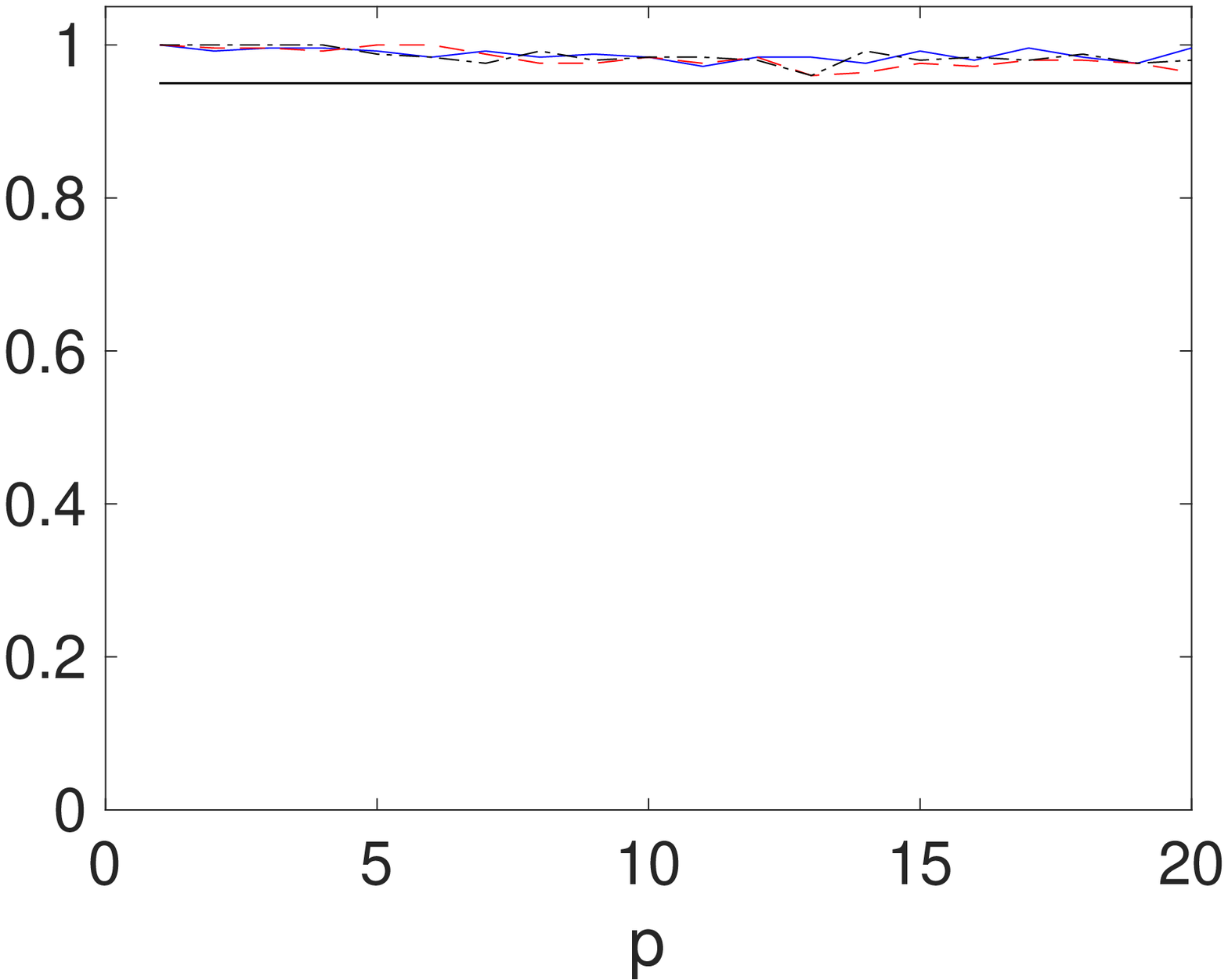}} &
	{\includegraphics[width=6.75cm]{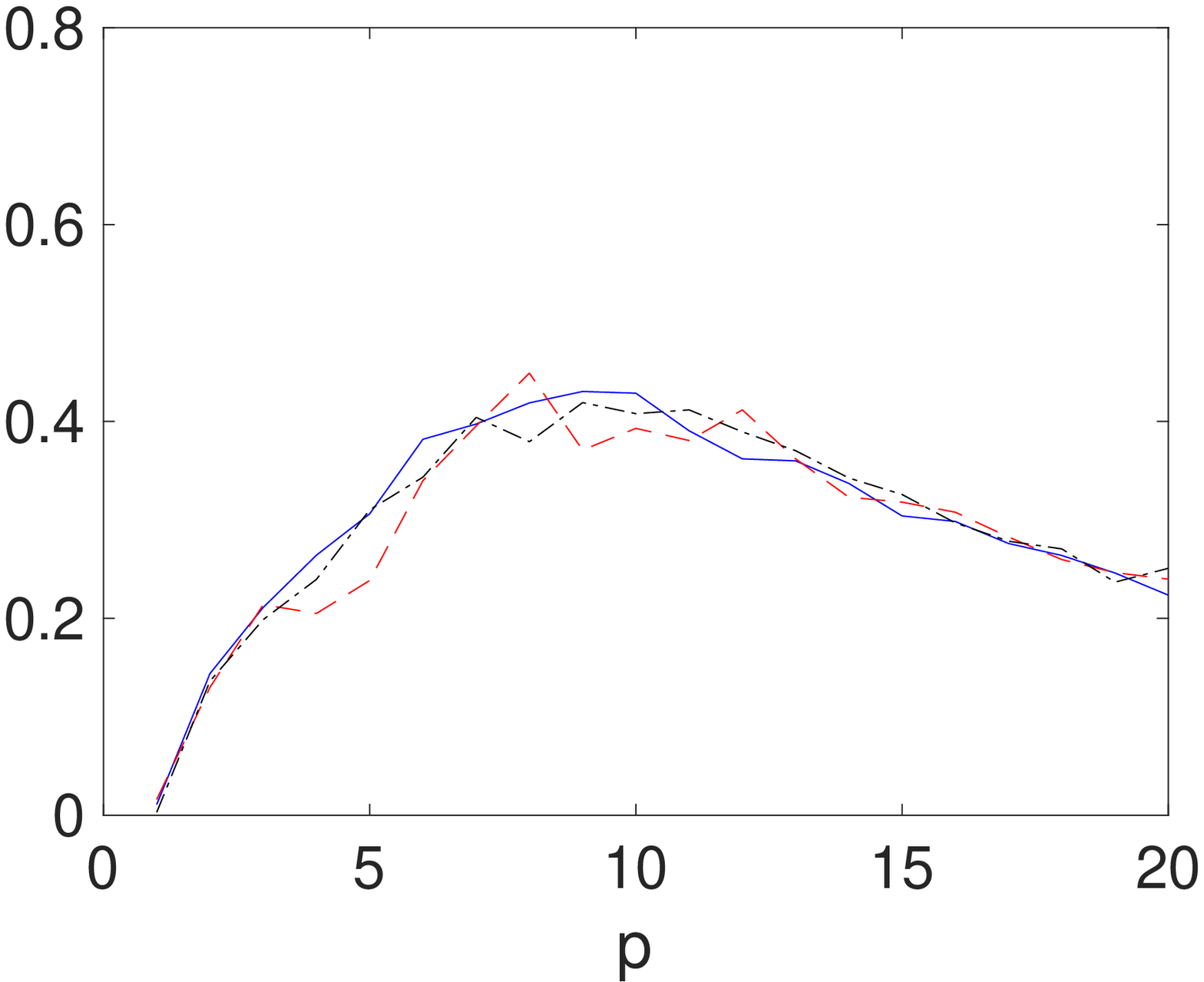}} \\
\end{tabular}
   	\caption{Frequentist coverage of the $95\%$ posterior credible interval for $p$ (left) and square root of relative mean squared error of the estimator of $p$ (right) for the SEPD. The simulations are for $\alpha=0.2$ (blue continuous line), $\alpha=0.5$ (red dashed line) and $\alpha=0.8$ (black dotted line), and for $T=100$ (top), $T=250$ (bottom).}
\label{FigCov_SEPD}
\end{figure}

As the yielded posterior distribution for the parameters is not analytically tractable, it is necessary to adopt Markov Chain Monte Carlo (MCMC) methods.  In particular, we have implemented a  Metropolis within Gibbs sampler. For each of the above $250$ samples, we have run $20000$ iterations of the MCMC algorithm and discarded the first $5000$ iterations as burn-in period. The results of the frequentist analysis of the posterior of $p$ are plotted in Figure \ref{FigCov_SEPD}. Examining the coverage, we note that the samples with $T=100$ have a frequency closer to the nominal value (i.e. 95\%) compared to the samples with $T=250$; this is more obvious for relatively large values of $p$. The MSE behaves in line with other frequentist studies for tail parameters (such as for the Student-$t$ and the skewed Student-$t$), with a smaller index value for larger sample size (as expected). Finally, we note that the effect of $\alpha$ on the frequentist performances is negligible.

To have a feeling of the complete inferential procedure, we show how all the parameters of a model are estimated. {In particular, we consider an autoregressive model with one lag, $\phi_1=-0.5$. The error terms are assumed to have an SEPD with  $\sigma=1$, $\alpha=0.23$ and $p=9$. We have drawn a sample of size $T=300$ from the model and implemented the MCMC procedure described above. In Figure \ref{FigCov_AEPD_special} we show the posterior chain and histogram for parameters $\alpha$, $\phi_1$, $p$ and $\sigma$. The corresponding posterior mean, median and 95\% credible interval are reported in Table \ref{table_AEPD_spec}. We note that the true parameter values are well contained in the corresponding posterior credible interval.
\begin{figure}[h!]
\centering
\begin{tabular}{cc}
	{\includegraphics[width=6.75cm]{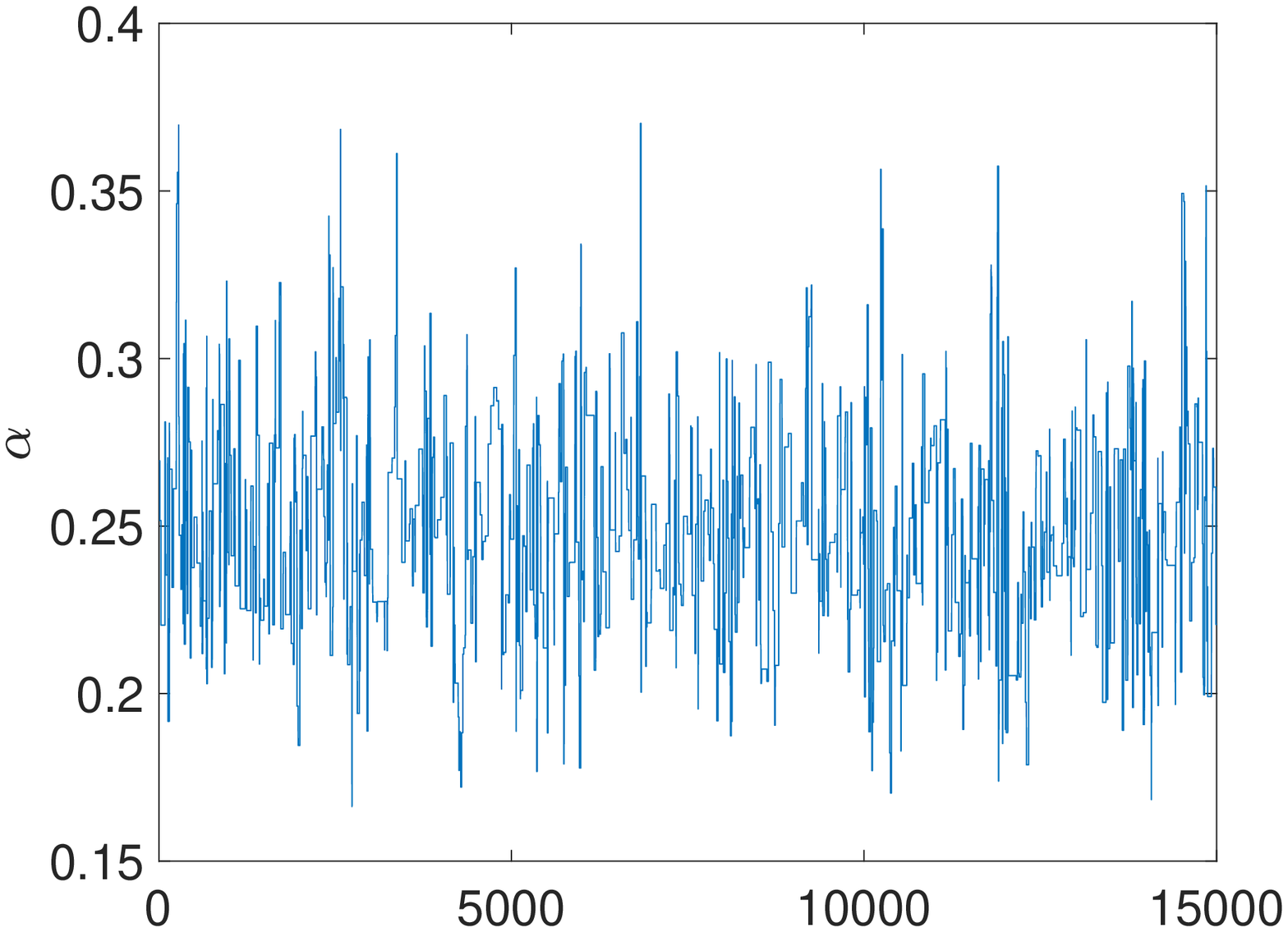}} &
	{\includegraphics[width=6.75cm]{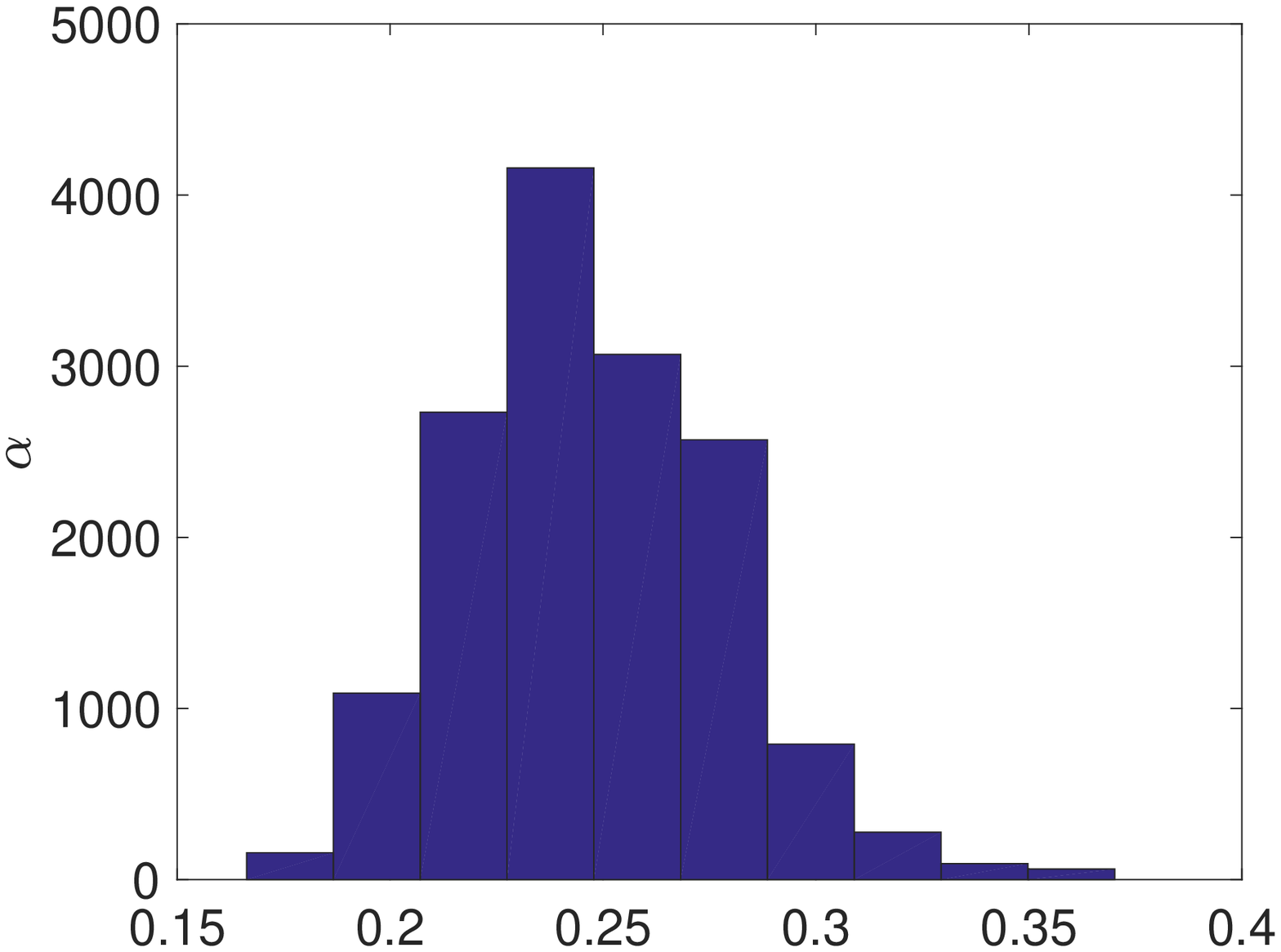}} \\
	{\includegraphics[width=6.75cm]{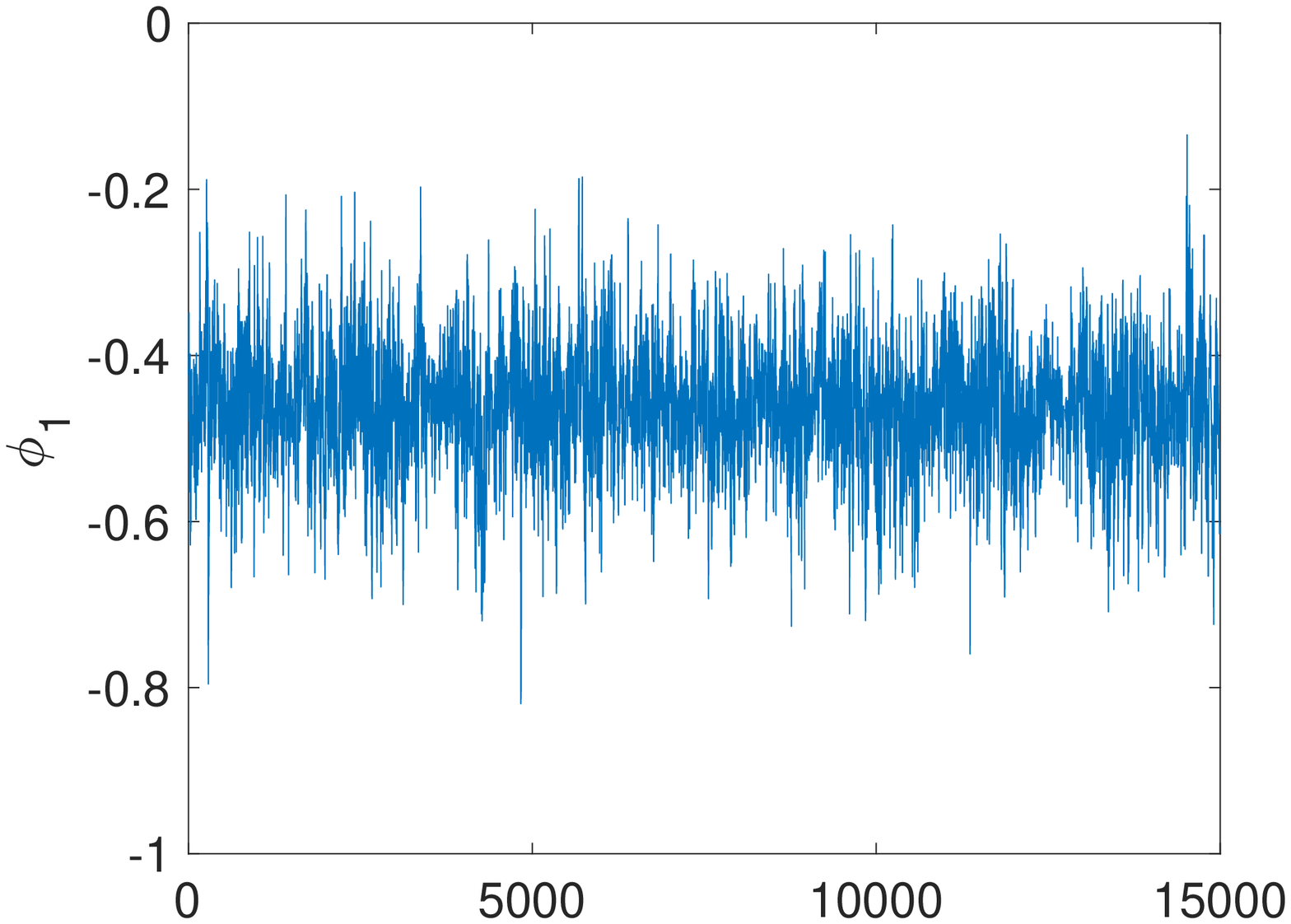}} &
	{\includegraphics[width=6.75cm]{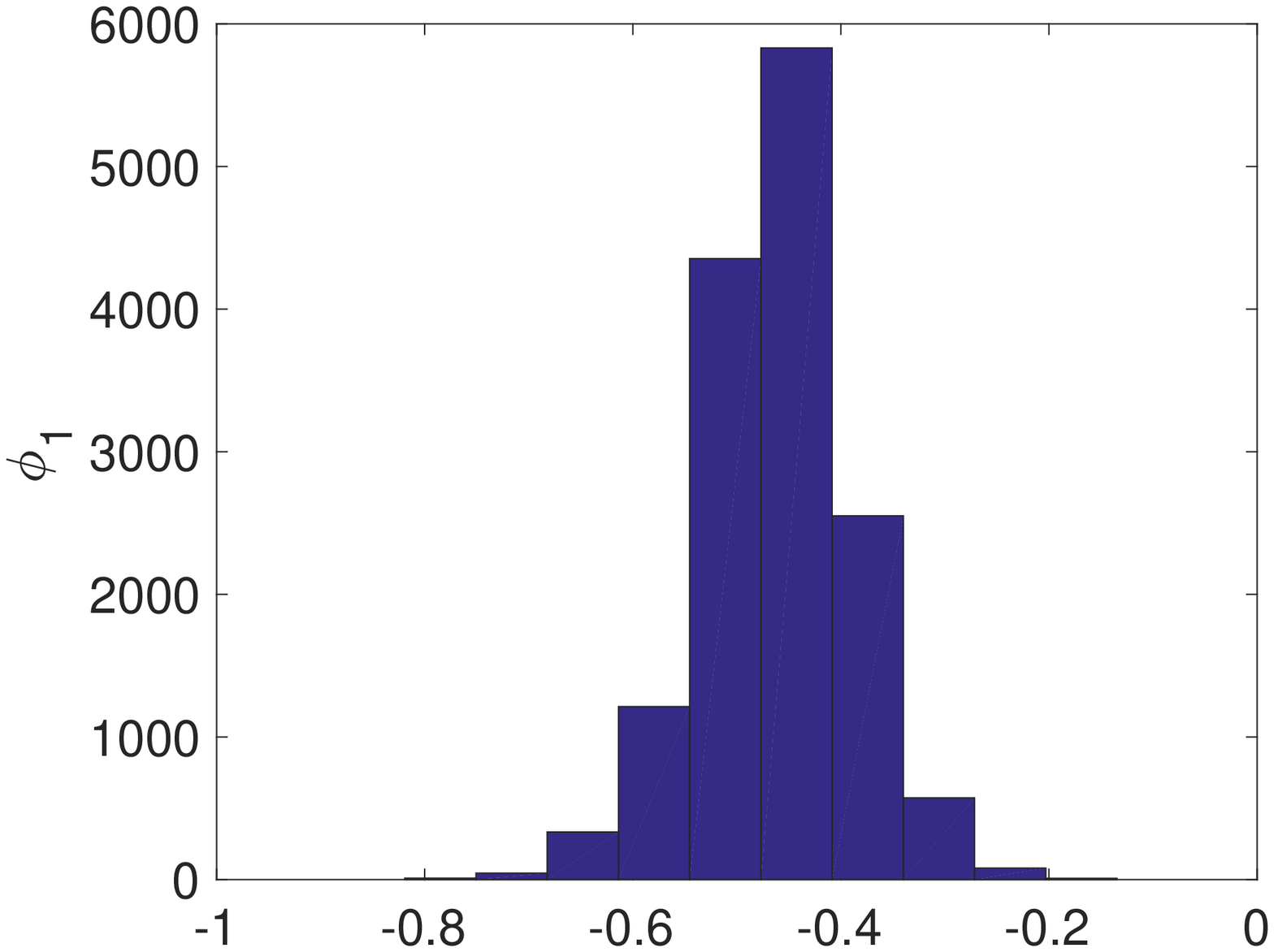}} \\
	{\includegraphics[width=6.75cm]{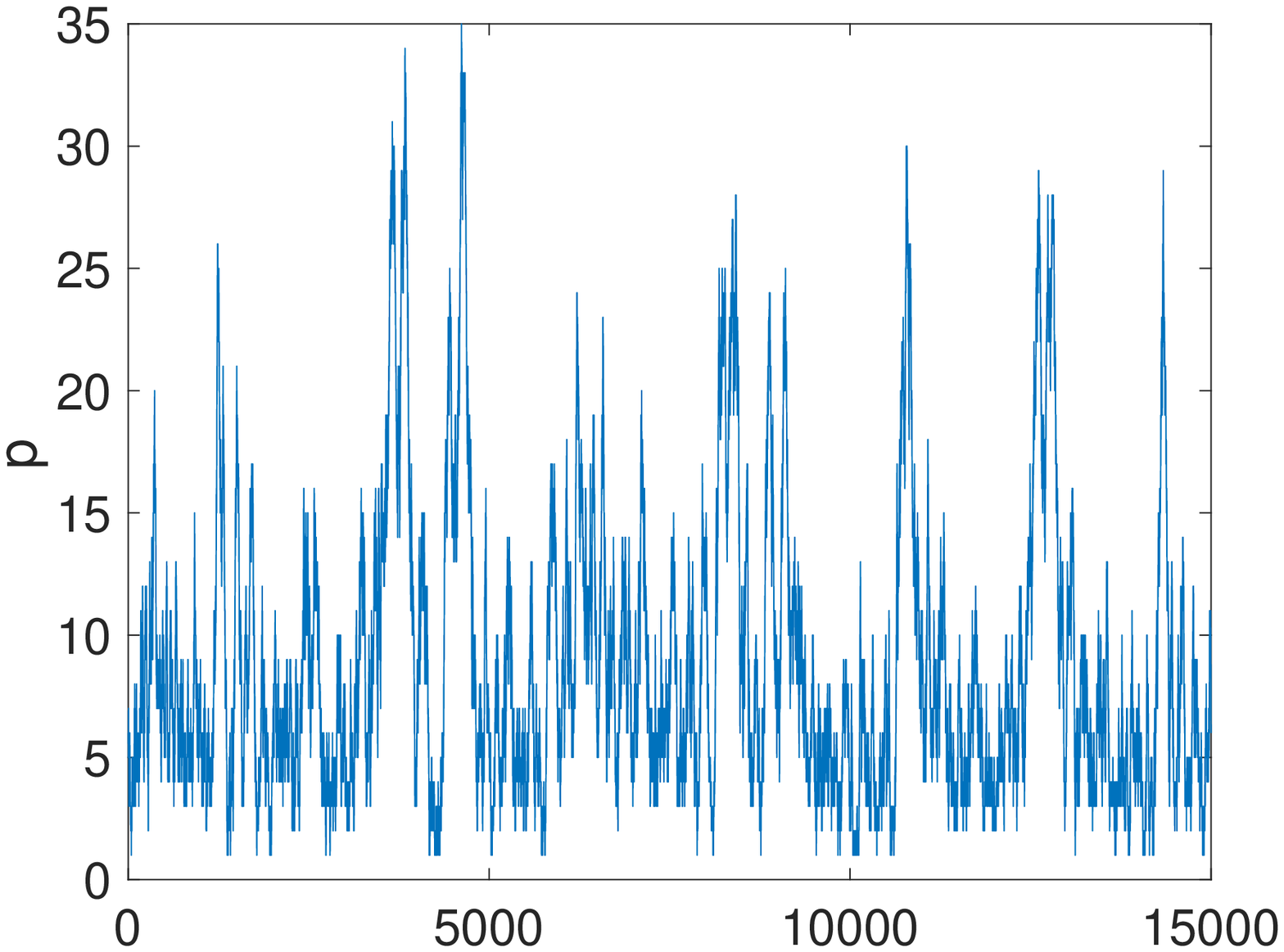}} &
	{\includegraphics[width=6.75cm]{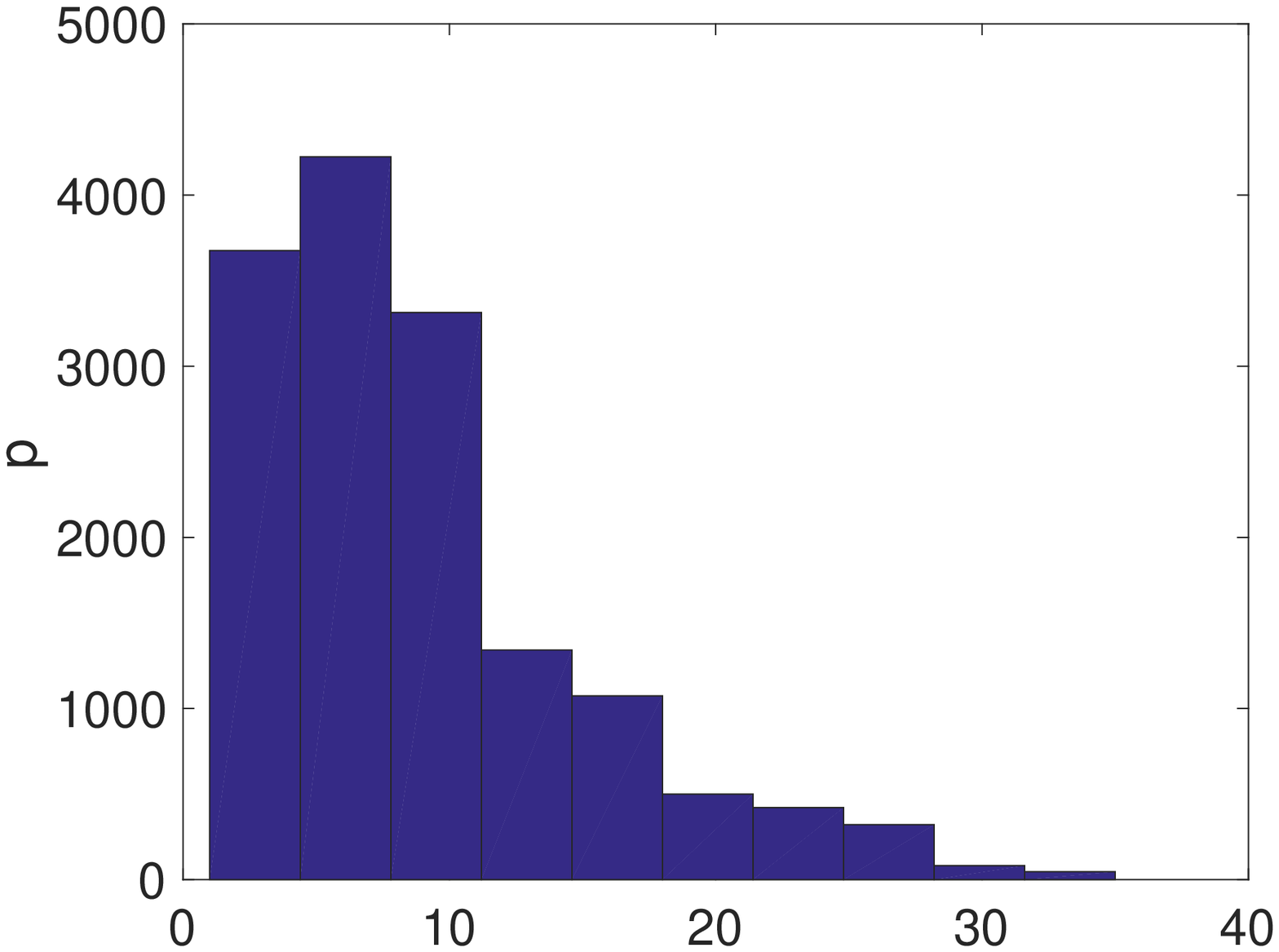}} \\
	{\includegraphics[width=6.75cm]{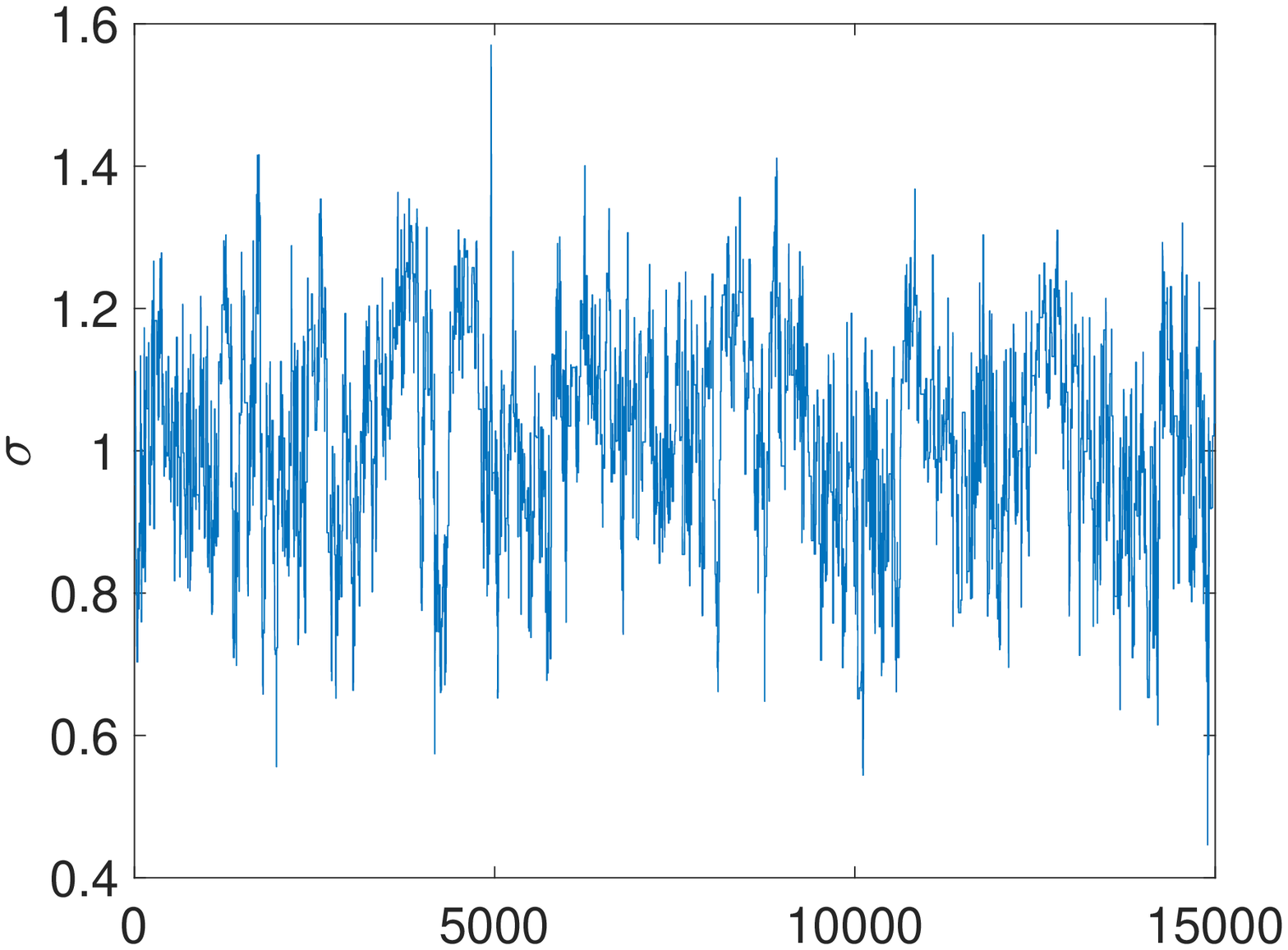}} &
	{\includegraphics[width=6.75cm]{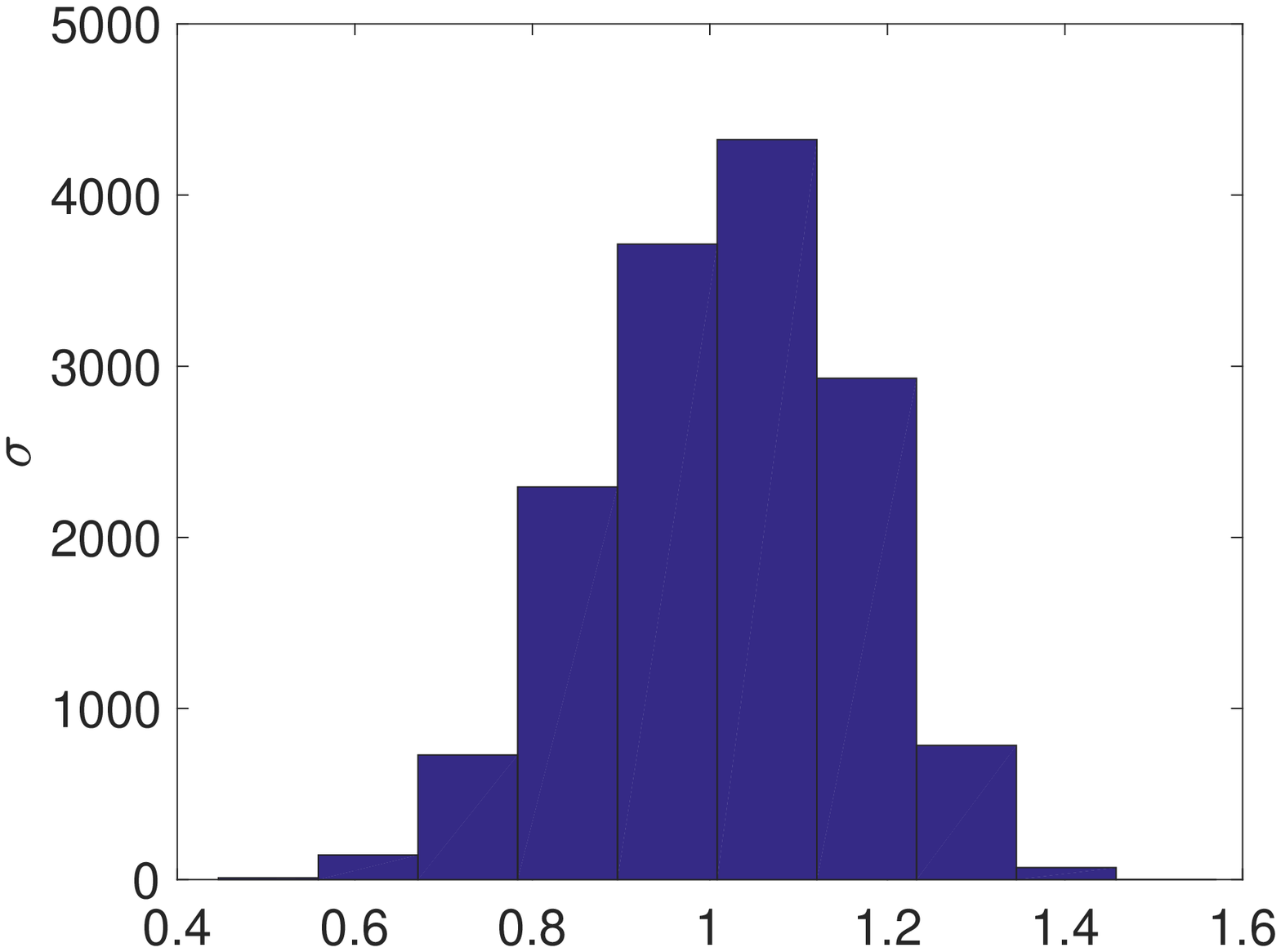}} \\
\end{tabular}
   	\caption{{Sample chains (left panels) and histograms of the posterior distributions (right panels) of the parameters for the simulated data from the SEPD with $\alpha = 0.23$, $\phi_1 = -0.5$, $p=9$, $\sigma = 1$ and $T=300$.}}
\label{FigCov_AEPD_special}
\end{figure}

\begin{table}[h!]
\centering
\begin{tabular}{cccc}
\hline
Parameter & Mean & Median & $95\%$ C.I. \\
\hline
$\alpha$ & 0.2476 & 0.2456 & (0.1917, 0.3125) \\
$\phi_1$ & -0.4612 & -0.4614 & (-0.6147, -0.3187) \\
$p$ & 8.94 & 8 & (2, 25) \\
$\sigma$ & 1.0179 & 1.0232 & (0.7104,1.2751) \\
\hline
\end{tabular}
\caption{{Summary statistics of the posterior distributions for the parameters of the simulated data from an SEPD with $\alpha = 0.23$, $\phi_1 = -0.5$, $p=9$, $\sigma =1$ and $T=300$.}}
\label{table_AEPD_spec}
\end{table}

\subsection{SGLD simulation study}
To study the performance of the loss-based prior for the tail parameter $p$ of the SGLD, as anticipated, we consider a linear regression model where the error terms have the above distribution. That is,
\begin{equation}
y_i =  \beta_0 + \beta_1 x_{i} + \varepsilon_i, \qquad i=1,\ldots,n, \label{Reg_SGLD}
\end{equation}
where, for the purpose of this simulation, we have set $\beta_0 = 1.5$, $\beta_1 = -1$ and $\varepsilon_i\sim SGLD(\alpha,p,\mu,\sigma)$. We select 250 random samples from the above model \eqref{Reg_SGLD} for each scenario determined by $p=1,\ldots,20$, $\alpha=\{0.3,0.5,0.8\}$ and $n=30,100$. The scale parameter $\sigma$ has been fixed to 1. The simulation study has been performed by considering a loss-based prior on $p$, the Jeffreys prior for the skewness parameter, $\pi(\alpha)\sim \mbox{Be}(1/2,1/2)$, and for the scale parameter, $\pi(\sigma)\propto 1/\sigma$ (as discussed in Section \ref{PrAMS}). For $\beta_0$ and $\beta_1$ we have used the Zellner g-prior \citep{Zellner86} with $g=n$, which is a bivariate normal with zero means and covariance matrix $\Sigma=n(X^\prime X)^{-1}$, where $X = (1, x_{1i})$.

For this model as well the posterior distribution is analytically intractable. Therefore, we have implemented an MCMC procedure (Metropolis within Gibbs samples) with 10000 iterations and a burn-in period of 5000 iterations. The frequentist analysis of the posterior for $p$ is shown in Figure \ref{FigCov_Logistic}. The coverage of the posterior 95\% credible interval appears to be very similar whether we consider the different values of the skewness parameter $\alpha$ or the sample size. For what it concerns the MSE, we note some differences when the sample size is 30, although these are most certainly due to the relatively small amount of information about $p$ contained in the sample. This difference vanishes for $n=100$.
\begin{figure}[h!]
\centering
\begin{tabular}{cc}
	{\includegraphics[width=6.75cm]{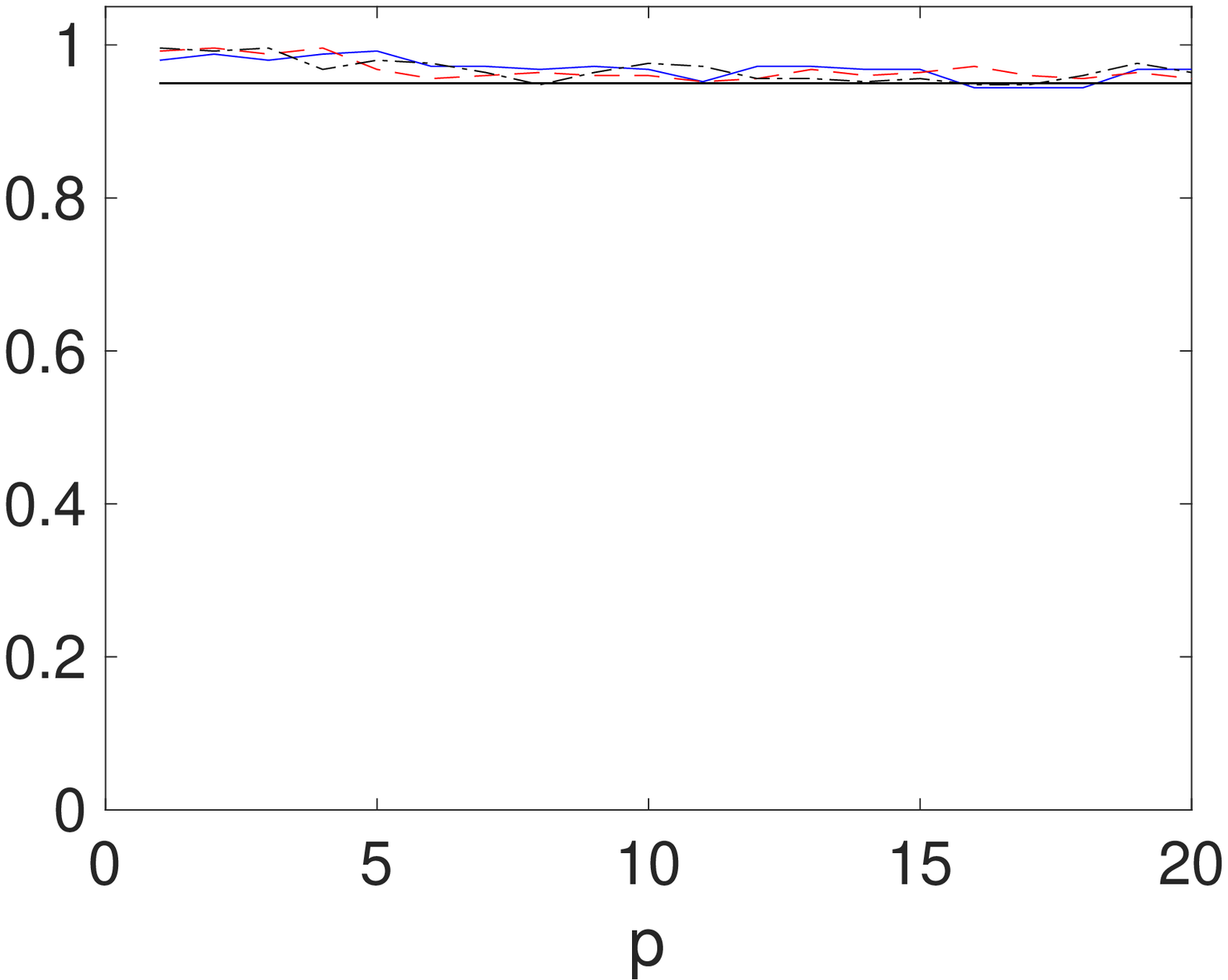}} &
	{\includegraphics[width=6.75cm]{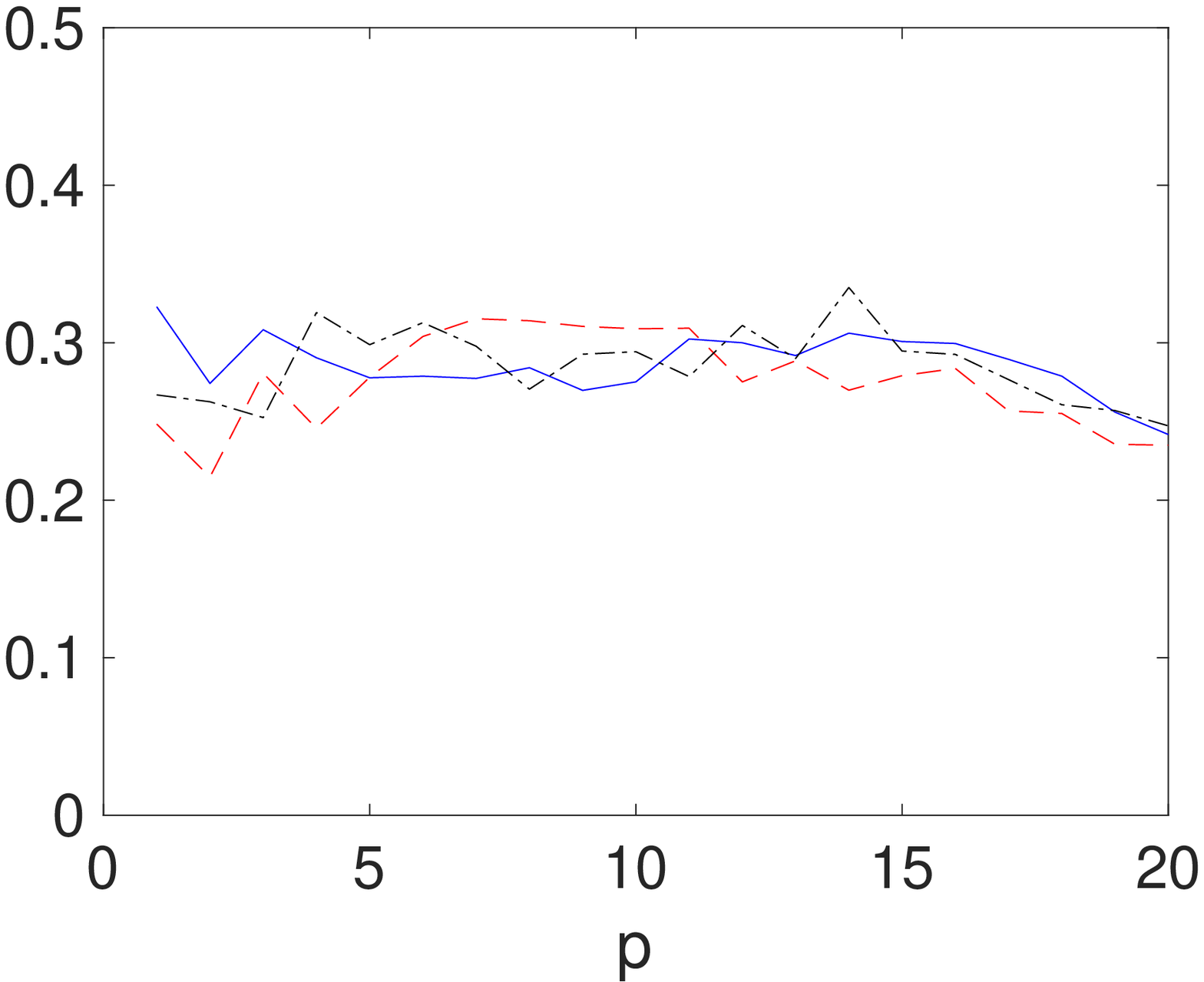}} \\
	{\includegraphics[width=6.75cm]{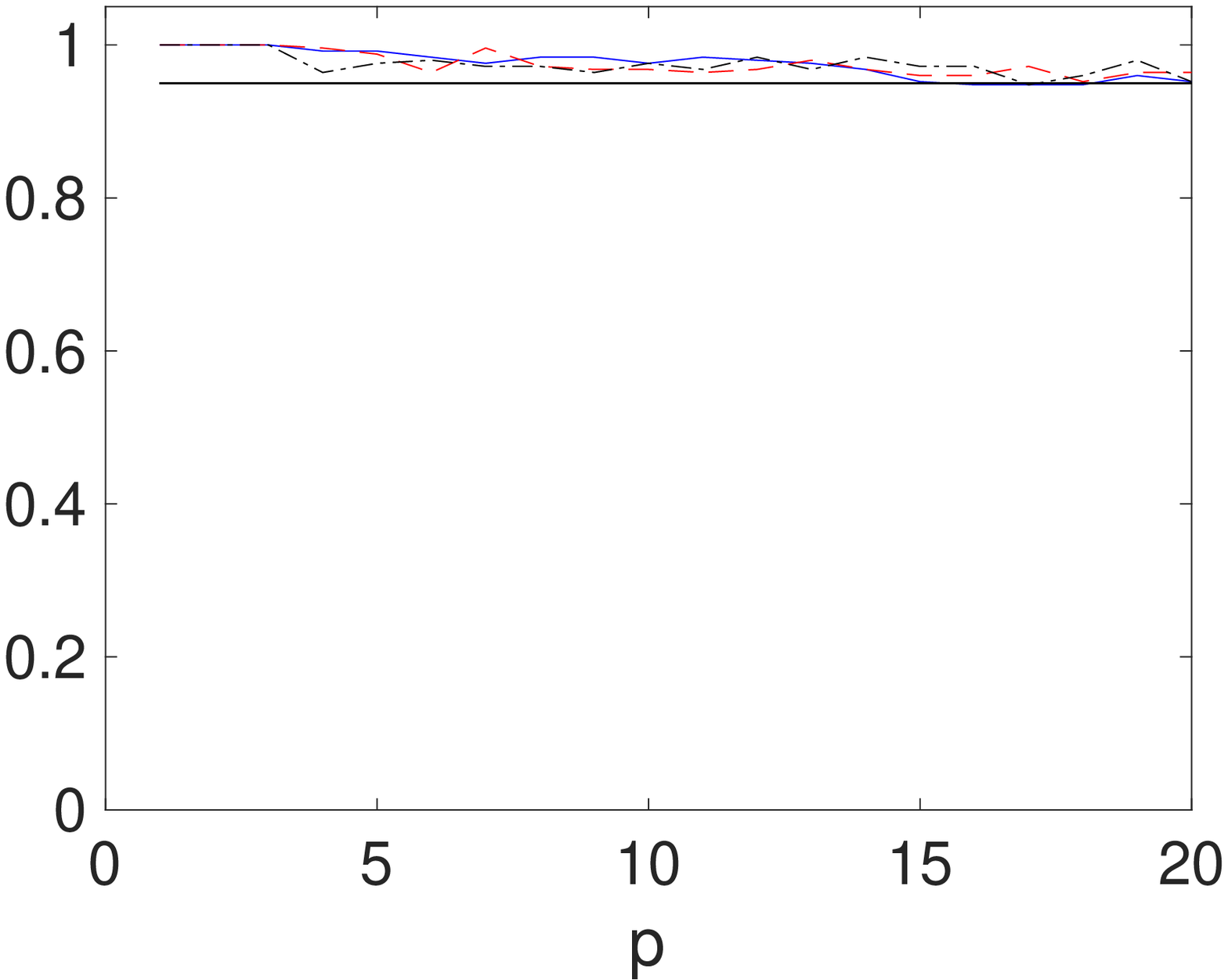}} &
	{\includegraphics[width=6.75cm]{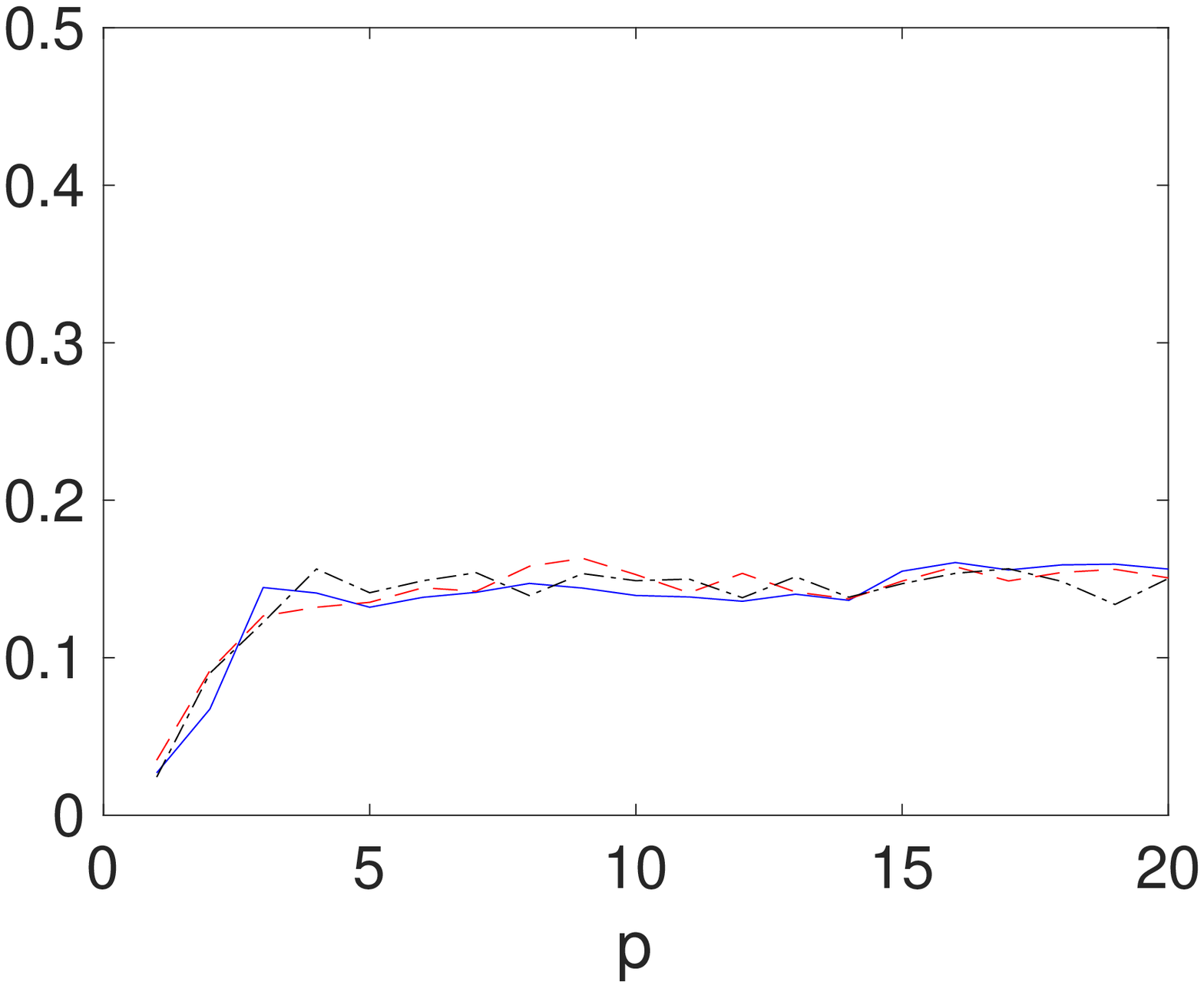}} \\
\end{tabular}
   	\caption{Frequentist coverage of the $95\%$ credible intervals for $p$ (left) and square root of relative mean squared error of the estimator of $p$ (right) for the SGLD. The simulations are for $\alpha=0.2$ (blue continuous line), $\alpha=0.5$ (red dashed line) and $\alpha=0.8$ (black dotted line), and for $n=30$ (top), $n=100$ (bottom).}
\label{FigCov_Logistic}
\end{figure}
{Similarly to the study of the SEPD model, we report the complete inferential procedure for a single sample drawn from the model in \eqref{Reg_SGLD}, where we have set $\beta_0=-2.5$, $\beta_1=3$, $\alpha=0.23$, $p=9$ and $\sigma=1$. We run an MCMC procedure with $30.000$ iterations and a burn-in period of $5.000$ iterations. The posterior chains and histograms are plotted in Figure \ref{FigCov_Logistic_special}, with the corresponding posterior statistics reported in Table \ref{table_SLD_spec}. We note that the posterior means and medians give an excellent point representation of the true parameter values, and that the posterior credible intervals contain the above true values giving a high level accuracy of the estimates.}
\begin{table}[h!]
\centering
\begin{tabular}{cccc}
\hline
Parameter & Mean & Median & $95\%$ C.I. \\
\hline
$\alpha$ & 0.1363 & 0.136 & (0.1155,0.1615) \\
$\beta_0$ & -2.5396 & -2.54 & (-2.5775, -2.4964) \\
$\beta_1$ & 3.0048 & 3.0041 & (2.9639, 3.0494) \\
$p$ & 8.9295 & 9 & (8, 10) \\
\hline
\end{tabular}
\caption{{Summary statistics of the posterior distributions for the parameters of the simulated data from an SGLD with $\alpha = 0.13$, $\beta_0 = -2.5$, $\beta_1 = 3$, $p=9$ and $n=300$.}}
\label{table_SLD_spec}
\end{table}

\begin{figure}[h!]
\centering
\begin{tabular}{cc}
	{\includegraphics[width=6.75cm]{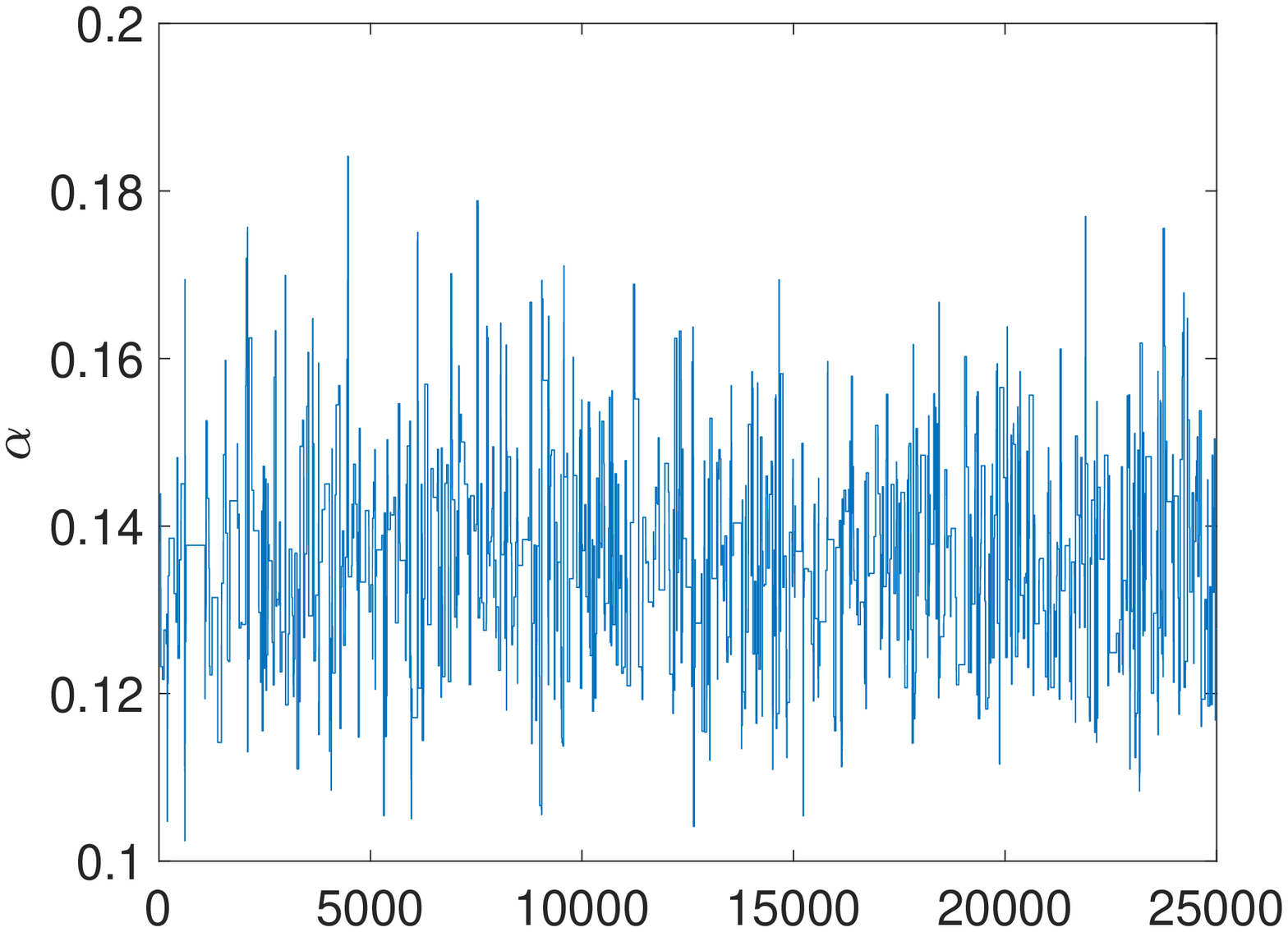}} &
	{\includegraphics[width=6.75cm]{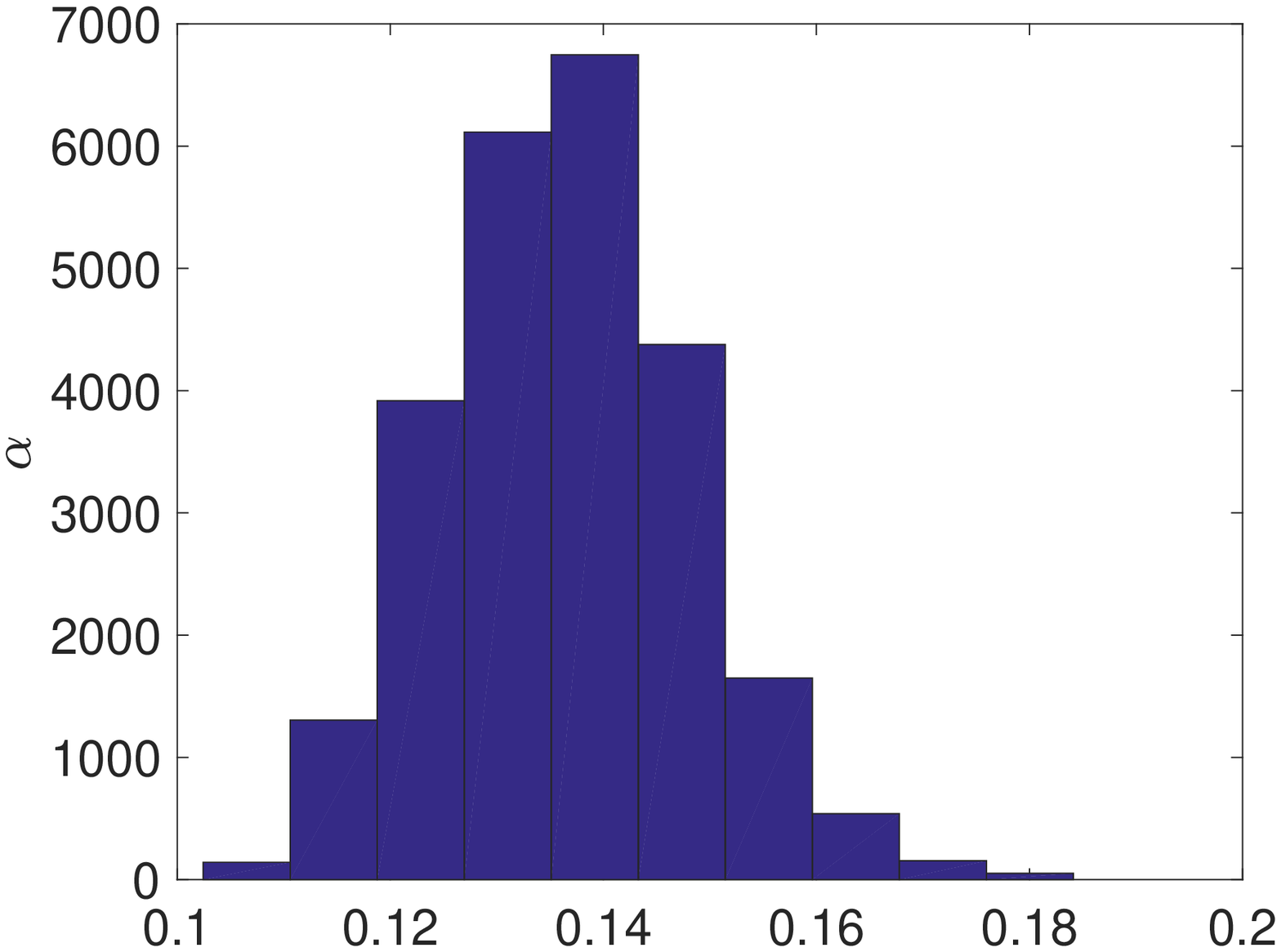}} \\
	{\includegraphics[width=6.75cm]{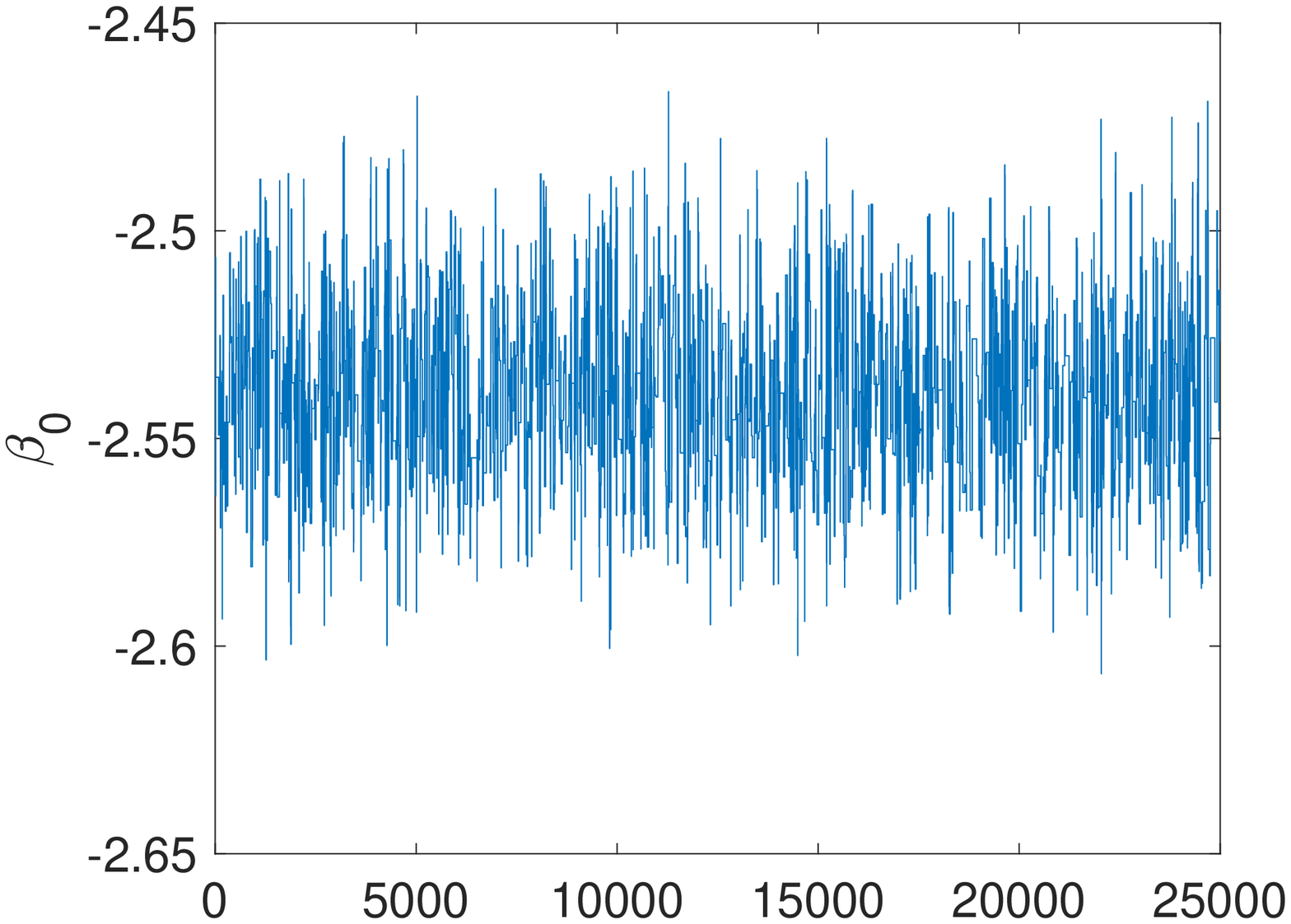}} &
	{\includegraphics[width=6.75cm]{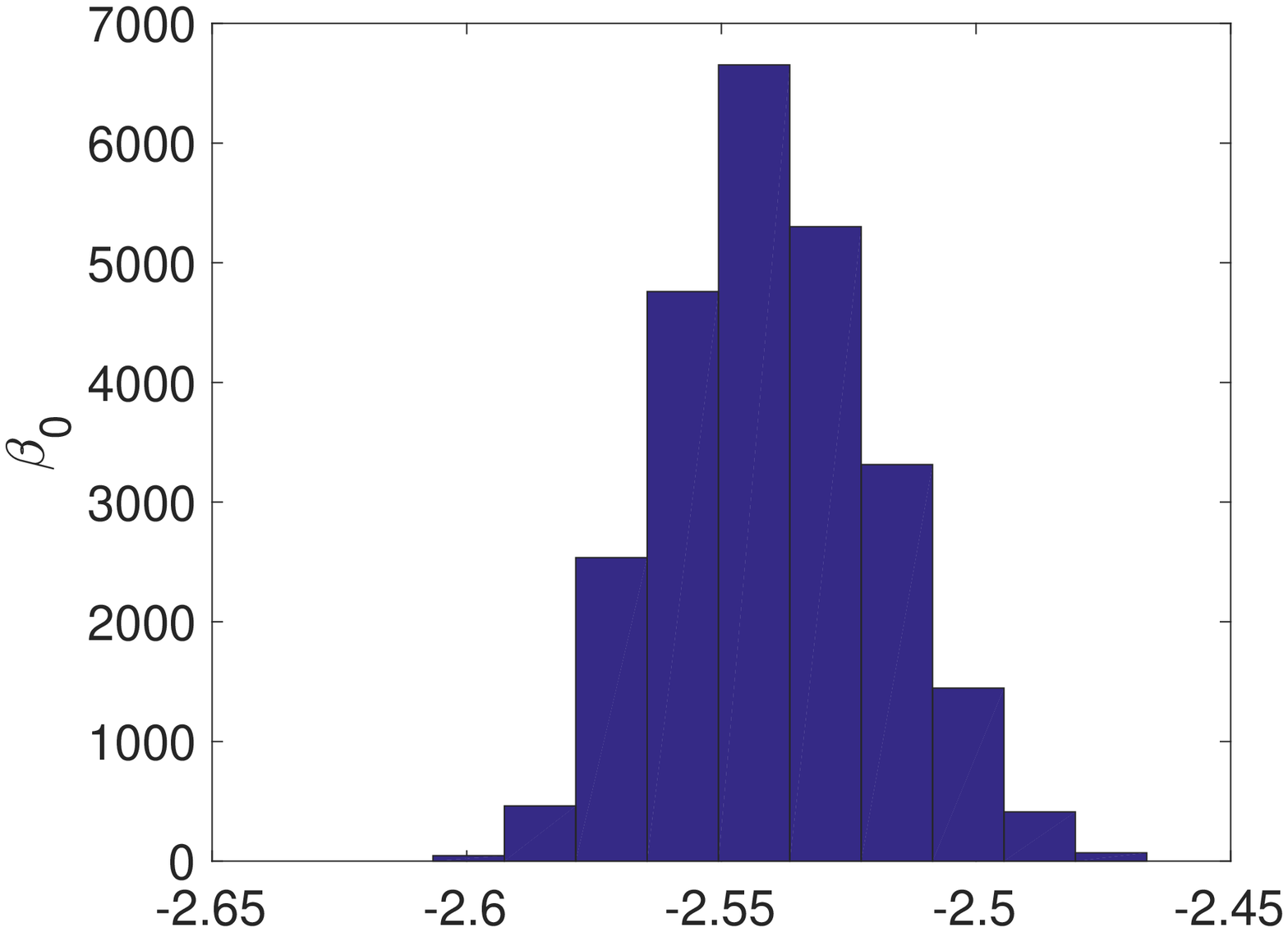}} \\
	{\includegraphics[width=6.75cm]{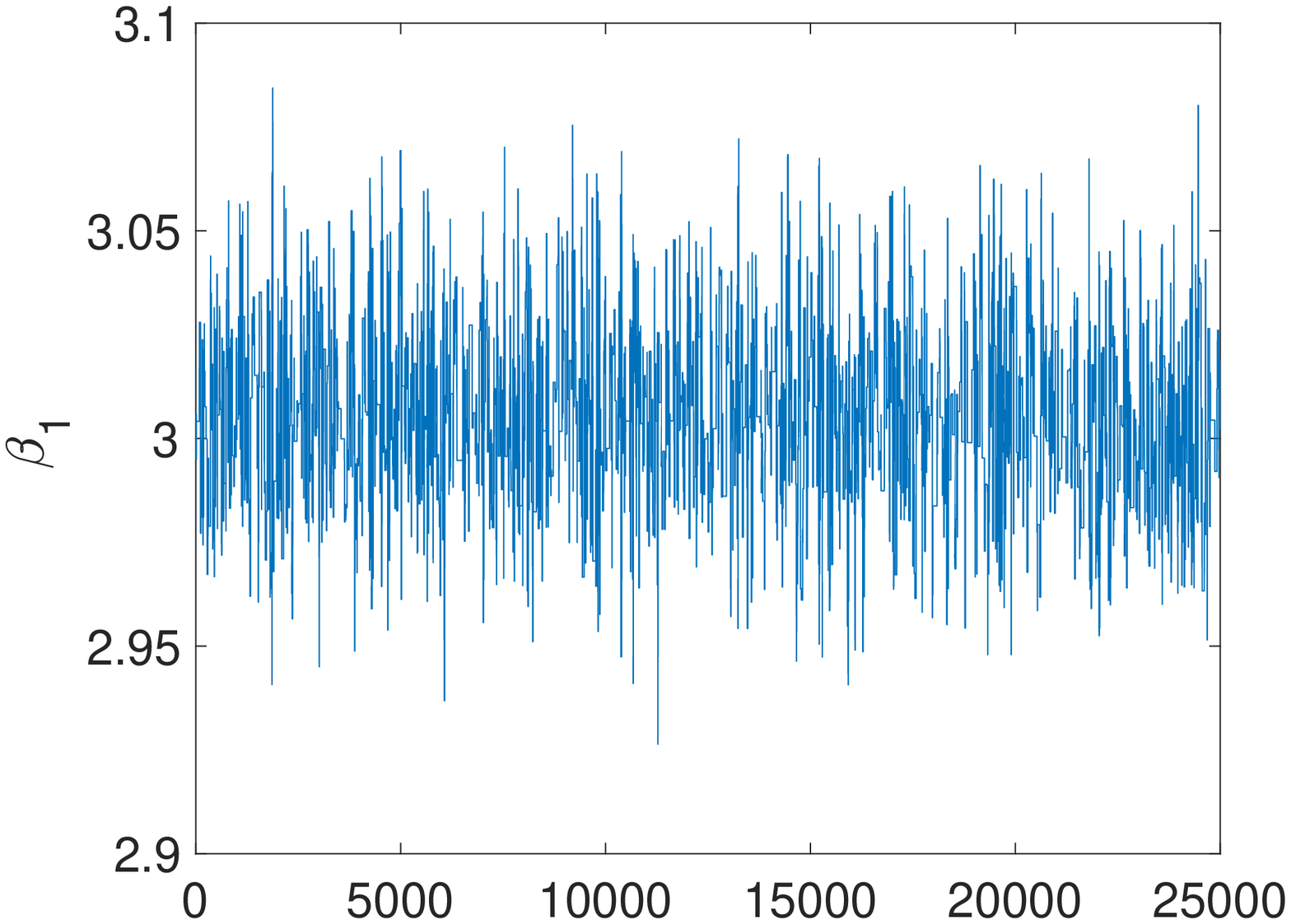}} &
	{\includegraphics[width=6.75cm]{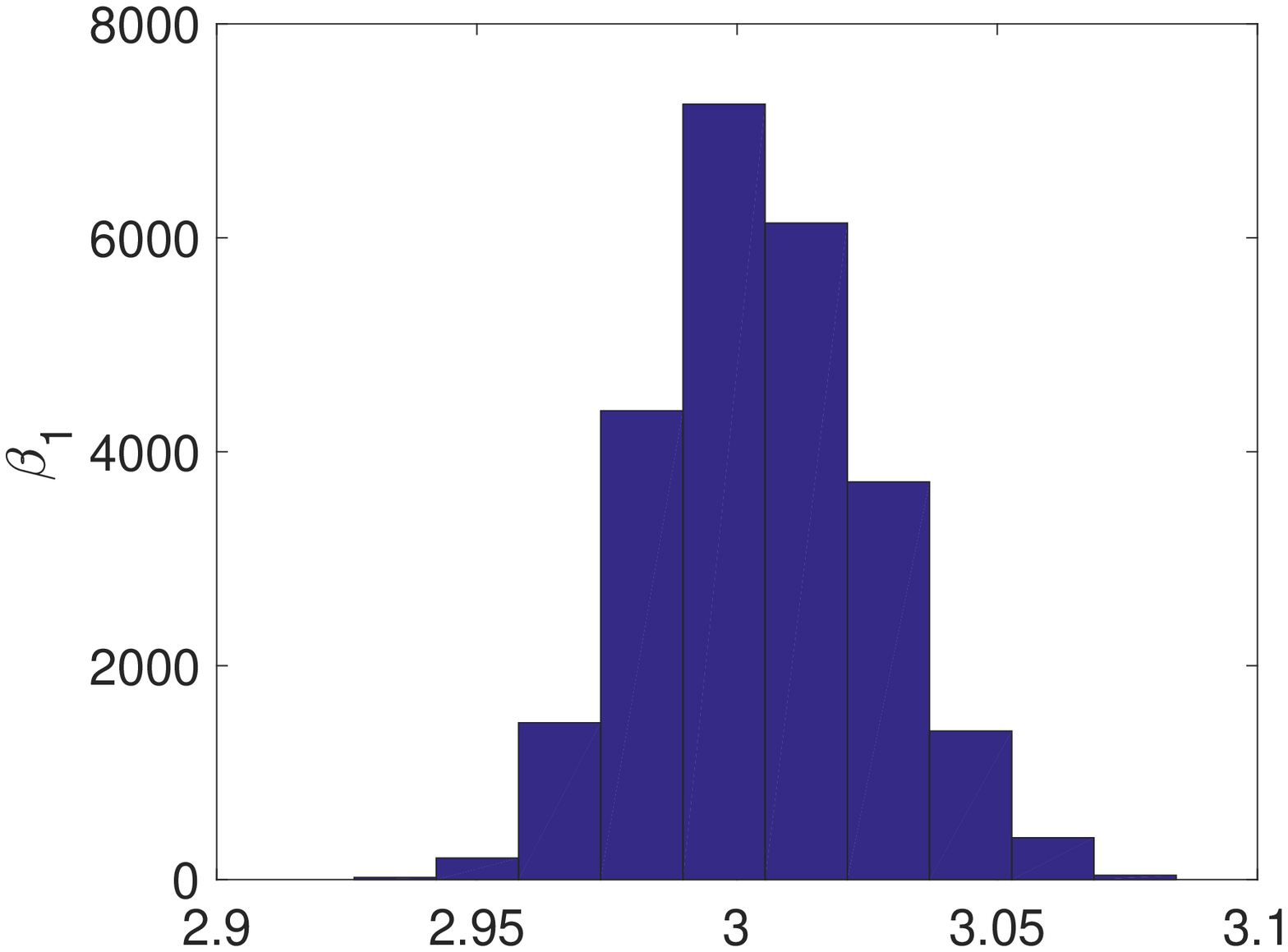}} \\
	{\includegraphics[width=6.75cm]{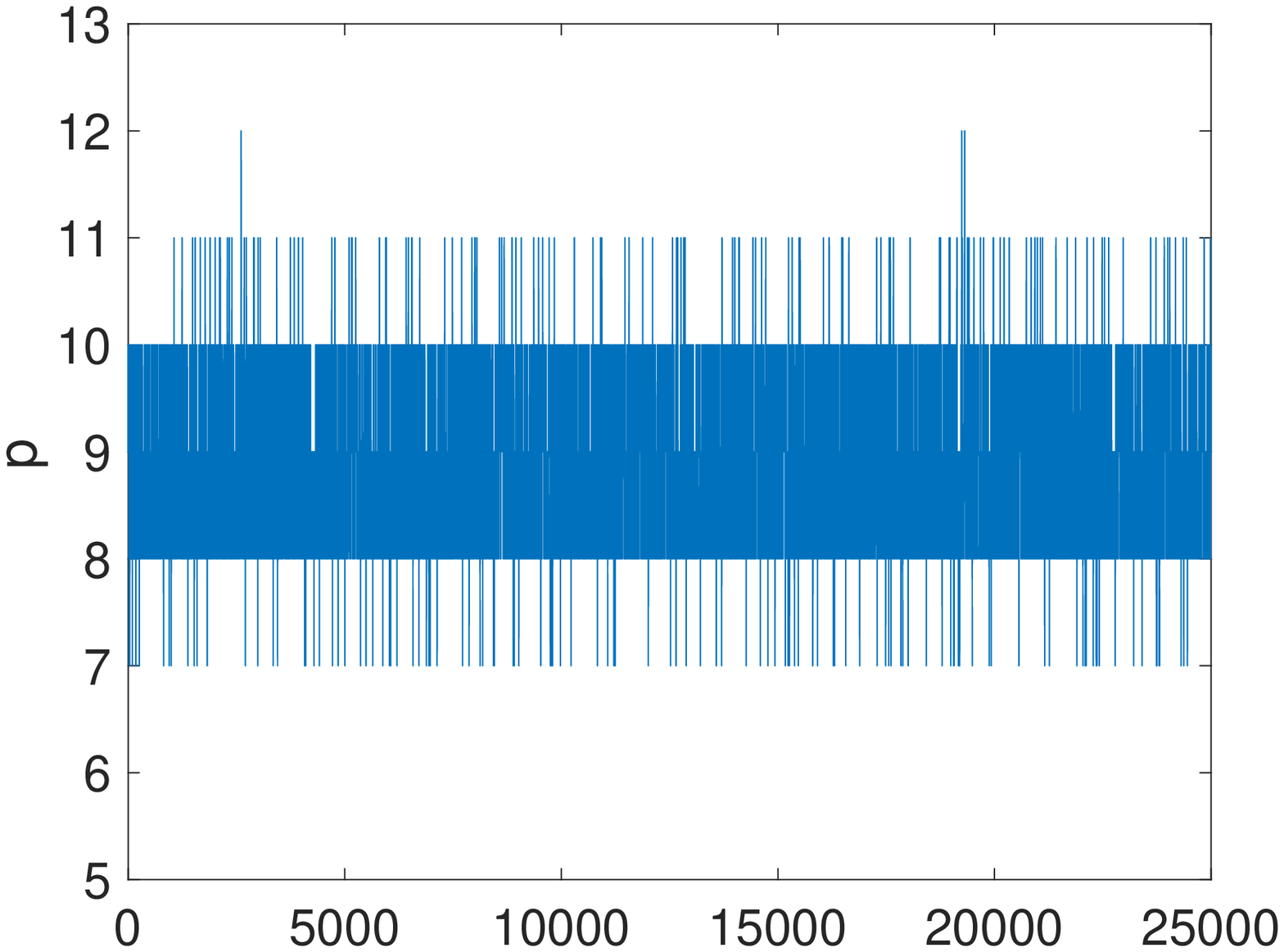}} &
	{\includegraphics[width=6.75cm]{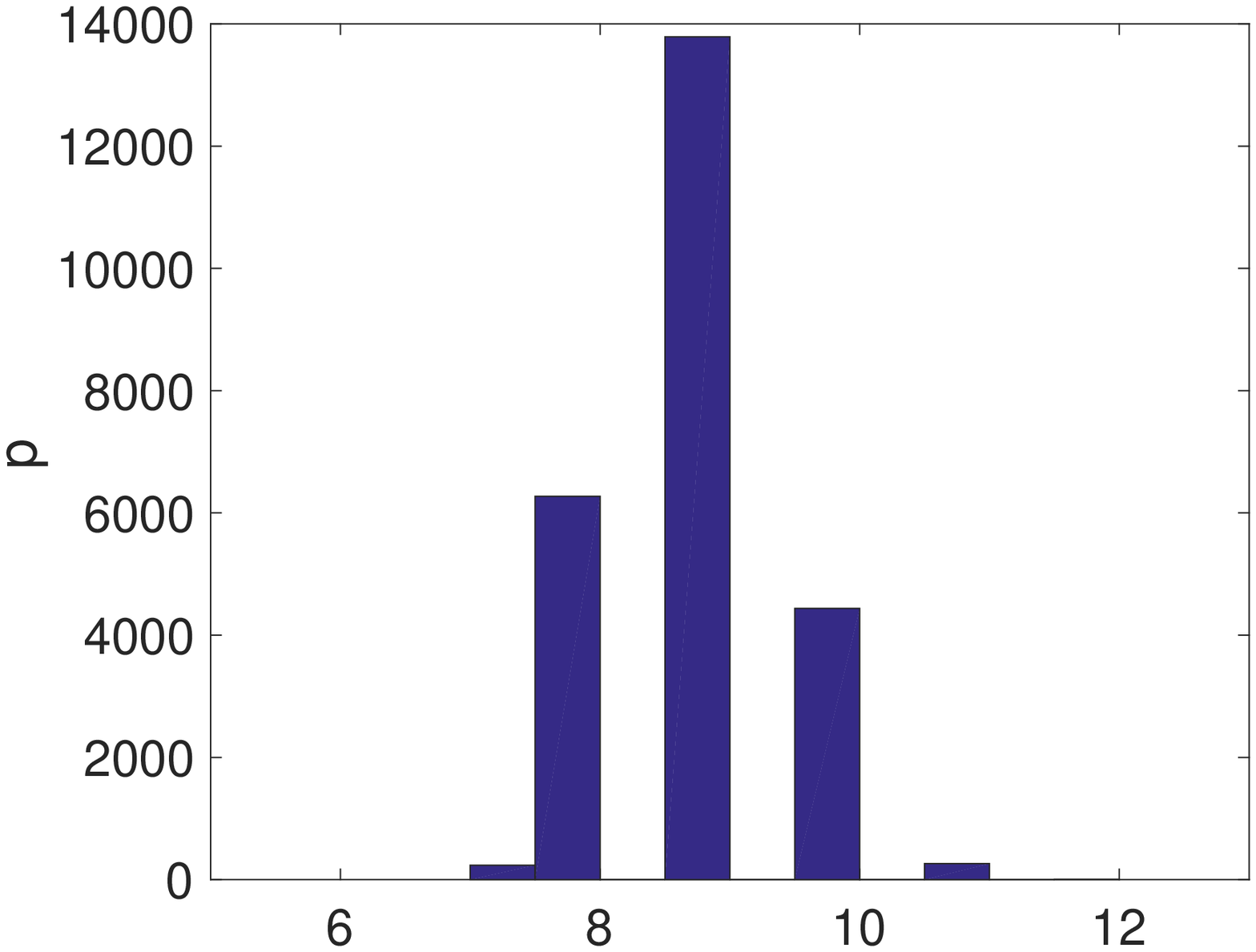}} \\
\end{tabular}
   	\caption{{ panels) of the parameters for the simulated data from the SGLD with $\alpha = 0.13$, $\beta_0 = -2.5$, $\beta_1 = 3$, $p=9$ and $n=300$.}}
\label{FigCov_Logistic_special}
\end{figure}


\section{Real Data Analysis}
\label{Real}
In this section, we present two different examples with publicly available data to illustrate how the loss-based prior for the tail parameter $p$ performs. In the first example we analyse the Nordpool Electricity prices by means of an autoregressive model with error terms distributed as a skewed exponential power, while in the second example we apply a linear regression model with error terms distributed as a skewed generalised logistic to Small Cell Cancer data.

\subsection{Nordpool Electricity Prices Data}

We use monthly prices (in level) to estimate models for electricity traded \citep{Bottazzi2011, Trindade2010} in Nordpool countries: in particular, Finland and Denmark. The prices, which have been obtained directly from the corresponding power exchanges, are plotted in Figure \ref{Fig:Prices}. Note that, for Denmark, we have averaged the two hourly zonal prices from Nordpool. The data is considered as the growth rate, meaning that we model the standardised first differences.
\begin{figure}[h!]
\centering
\begin{tabular}{c}
	\includegraphics[scale=0.5]{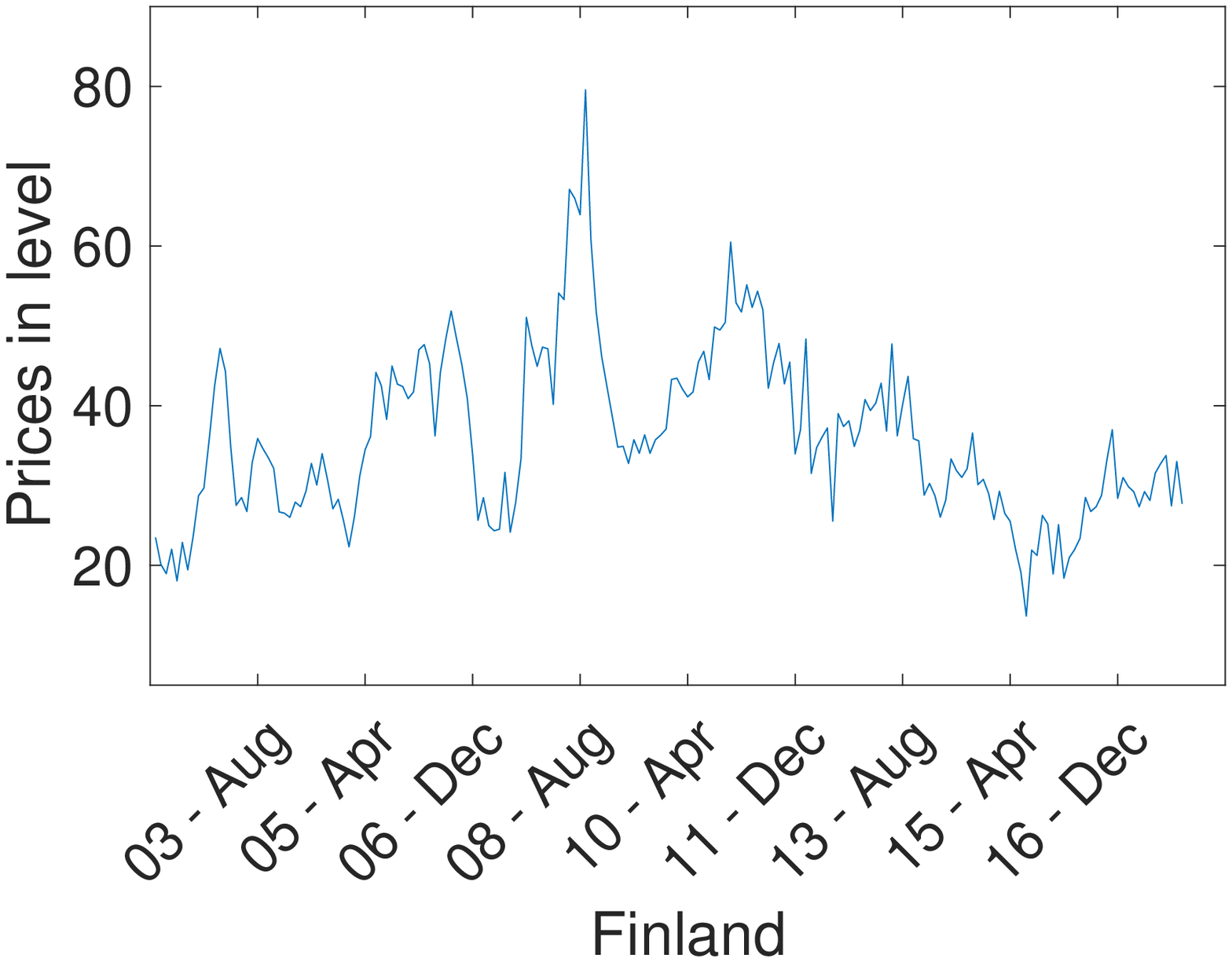} \\
	\includegraphics[scale=0.5]{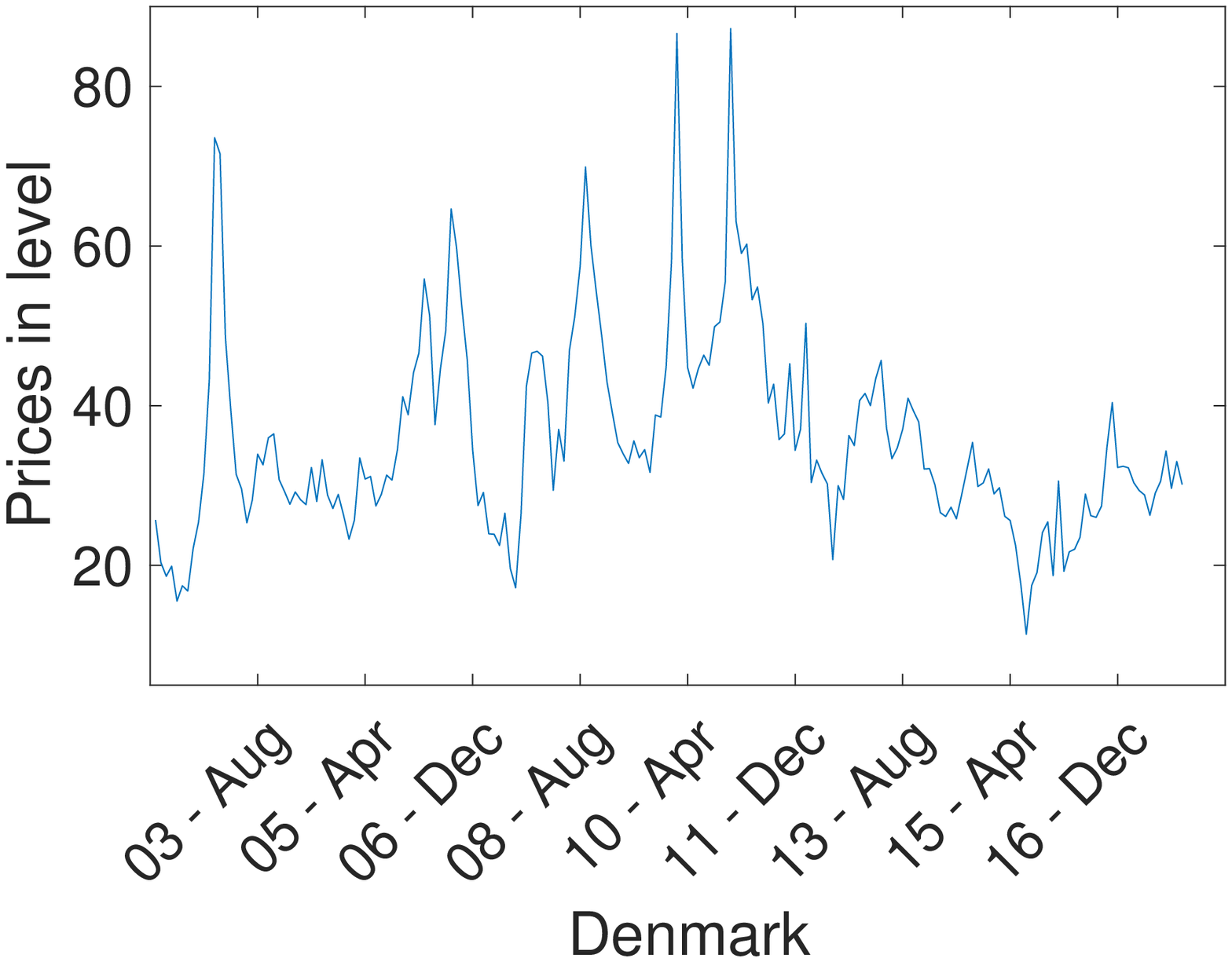} \\
\end{tabular}
\caption{Monthly electricity prices (in level) for Finland (top) and Denmark (bottom) from January 2003 to December 2017.}
\label{Fig:Prices}
\end{figure}
Finally, of the $T=180$ observation points, we use the first ten years as estimation sample and the last five years as forecast evaluation period.

The data is modelled with a univariate autoregressive model with one lag, where the error terms are SEPD. The results of the analysis are based on one-step-ahead forecasting process with a rolling window approach of 10 years for both the countries, and we have a forecast evaluation period of 60 observations (from January 2013 to December 2017). Following the results in Section \ref{Simu}, we run the estimation procedure through Gibbs sampling with a burn-in of $5.000$ iterations and for the forecasting procedure we use the remaining $15.000$ iterations.

We assess the goodness of our forecasts using different point and density metrics. For point forecasts, we use the root mean square errors (RMSEs) for the monthly prices as follows:
\begin{equation}
\mbox{RMSE} = \sqrt{\frac{1}{T-R} \sum_{t=R}^{T-1} \left(\hat{y}_{t+1|t}-y_{t+1|t}\right)^2}, \label{RMSE}
\end{equation}
where $T$ is the number of observations, $R$ is the length of the rolling window and $\hat{y}_{t+1|t}$ are the price forecasts.

To evaluate density forecasts, we use both the average log predictive score and the average continuous ranked probability score (CRPS). The log predictive score is computed as follows (see \cite{Geweke2010})
\begin{equation}
s_t(y_{t+1}) = \log{\left(f(y_{t+1}|I_t\right)}, \notag
\end{equation}
where $f(y_{t+1}|I_t)$ is the predictive density for $y_{t+1}$ constructed using information up to time $t$. In addition, following \cite{Gneiting2007} and \cite{GneitingRanjan2011}, we also compute the continuous ranked probability score, which has some advantages with respect to the log-score. In fact, it is less sensitive to outliers. It can be computed as follows:
\begin{equation}
\mbox{CRPS}_t(y_{t+1}) = \int_{-\infty}^{\infty} \left(F(z) - \mathbb{I}\{y_{t+1}\le z\}\right)^2 dz = E_f|Y_{t+1}-y_{t+1}| - 0.5 E_f|Y_{t+1}- Y'_{t+1}|,
\end{equation}
where $F$ denotes the cumulative distribution function associated with the predictive density $f$, $\mathbb{I}\{y_{t+1}\le z\}$ denotes an indicator function taking value $1$ if $y_{t+1}\le z$ and $0$ otherwise, and $Y_{t+1}$ and $Y'_{t+1}$ are independent random draws from the posterior predictive density. 

In Table \ref{tab_Forecast}, we report the RMSEs, average log-scores and average CRPS for the benchmark model, which is referred as the AR model with frequentist estimation. We compare the results from the Ordinary Least Squares (OLS) benchmark with the results obtained from the Bayesian AR with Normal error and our model based on SEPD errors. We also report the ratios of each model RMSE (average CRPS) to the baseline AR model, such that entries less than 1 indicate that the given model yields forecasts more accurate than those from the baseline. For the log-score, positive differences in score indicate that the given model outperforms the baseline.

\begin{table}[h!]
\centering
\begin{tabular}{ll|ccc}
\hline
 & Forecast & OLS & Bayesian Normal & SEPD error \\ 
\hline
\textit{Finland} & RMSE & 0.749 & 1.003 & 1.022 \\
 & log-score & -1.334 & 0.090 & 0.106 \\
 & CRPS & 0.491 & 0.899 & 0.891 \\
\hline
\textit{Denmark} & RMSE & 0.541 & 1.001 & 1.021 \\
 & log-score & -1.143 & 0.013 & 0.165 \\
 & CRPS & 0.355 & 0.990 & 0.918 \\
\hline
\end{tabular}
\caption{Point (RMSE) and density forecast (average log predictive score and average CRPS) for Finland and Denmark. {The first column (OLS) refers to the benchmark model and shows the values of the RMSE, average log predictive score and average CRPS. The second (Bayesian Normal) and third (SEPD error) columns refer to the RMSE ratios, score differences and CRPS ratios with respect to the benchmark model (OLS).}}
\label{tab_Forecast}
\end{table}

For both Finland and Denmark the point forecast appears to be worse than the benchmark. This is more obvious for the AR model with SEPD errors, although the values are not far from one. There is a noticeable improvement, in using SEPD errors, when we focus on density forecast. In fact, considering the log-score, we have a improvement in considering SEPD (instead of normal) errors from 0.090 to 0.106 for Finland, and a more obvious improvement from 0.013 to 0.165 for Denmark.

\subsection{Small Cell Cancer Data}
In this second example we illustrate the loss-based prior for the tail parameter $p$ when we employ a linear regression model with SGLD errors. The data has been obtained from \cite{Ying1995}, where a lung cancer study with two differente types of treatment has been performed. In particular, the study contained $n=121$ survival times (in log-days) of patients with small cell lung cancer (SCLC) to whom were administrated two different therapies. A treatment consisted of a combination of etoposide (E) and cisplatin (P) in any order. The patient were split into two treatment groups: treatment A ($62$ patients), where the therapy consisted in administering P followed by E; treatment B ($59$ patients), where the therapy consisted in administering E followed by P. We regress the survival time on the following two covariates: the entry age (in years) and a dummy variable identifying the type of treatment (A or B).

The estimation of the parameters of the regression model has been done through Monte Carlo methods, as described in Section \ref{Simu}, with 50000 iterations and a burn-in period of 10000 iterations. Figure \ref{Fig:cancer} shows the histograms of the posterior distributions for the parameters, while in Table \ref{tab_SCLC} we have the corresponding posterior statistics. 

\begin{figure}
\centering
\begin{tabular}{cc}
\includegraphics[width=6.75cm]{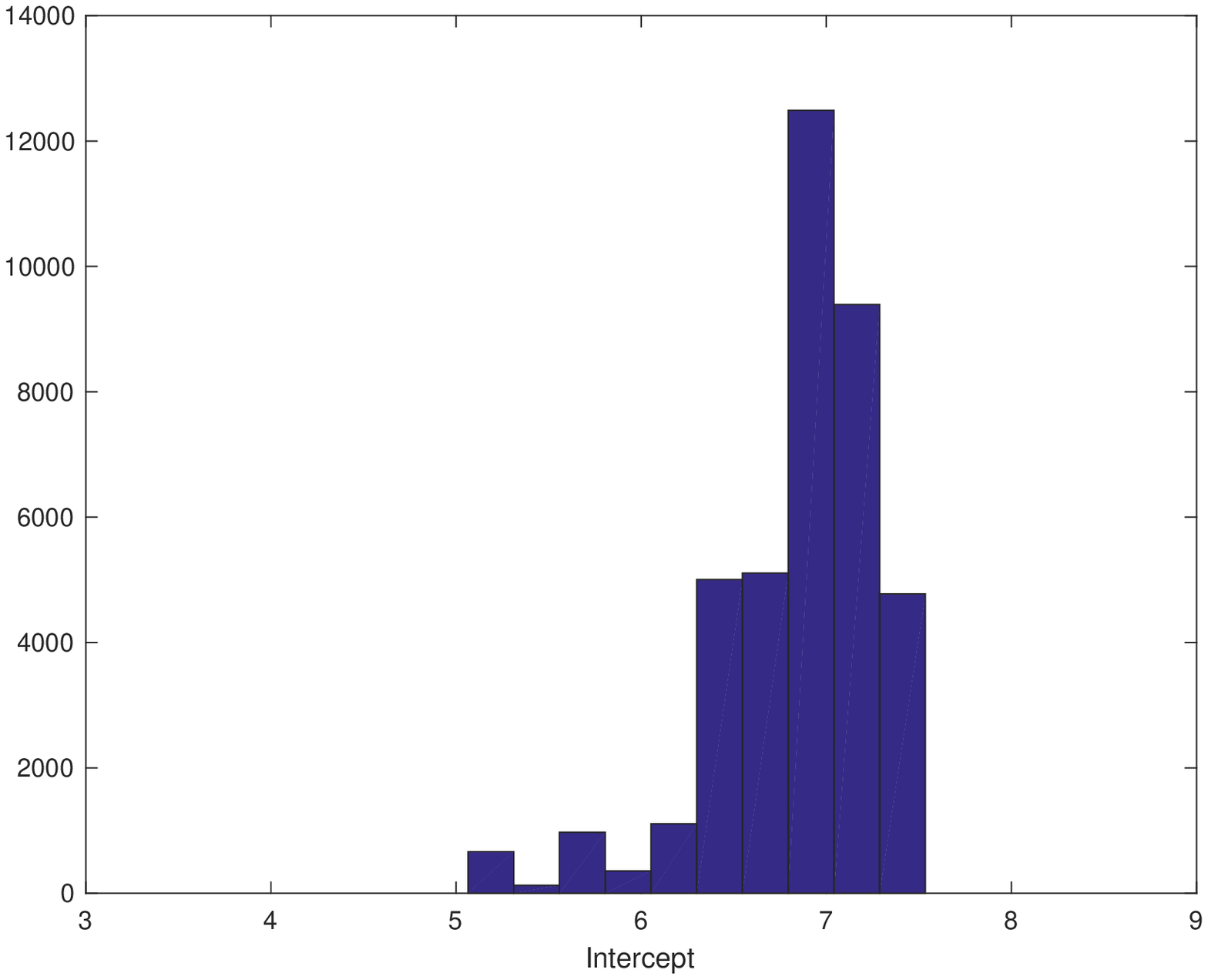} &
	\includegraphics[width=6.75cm]{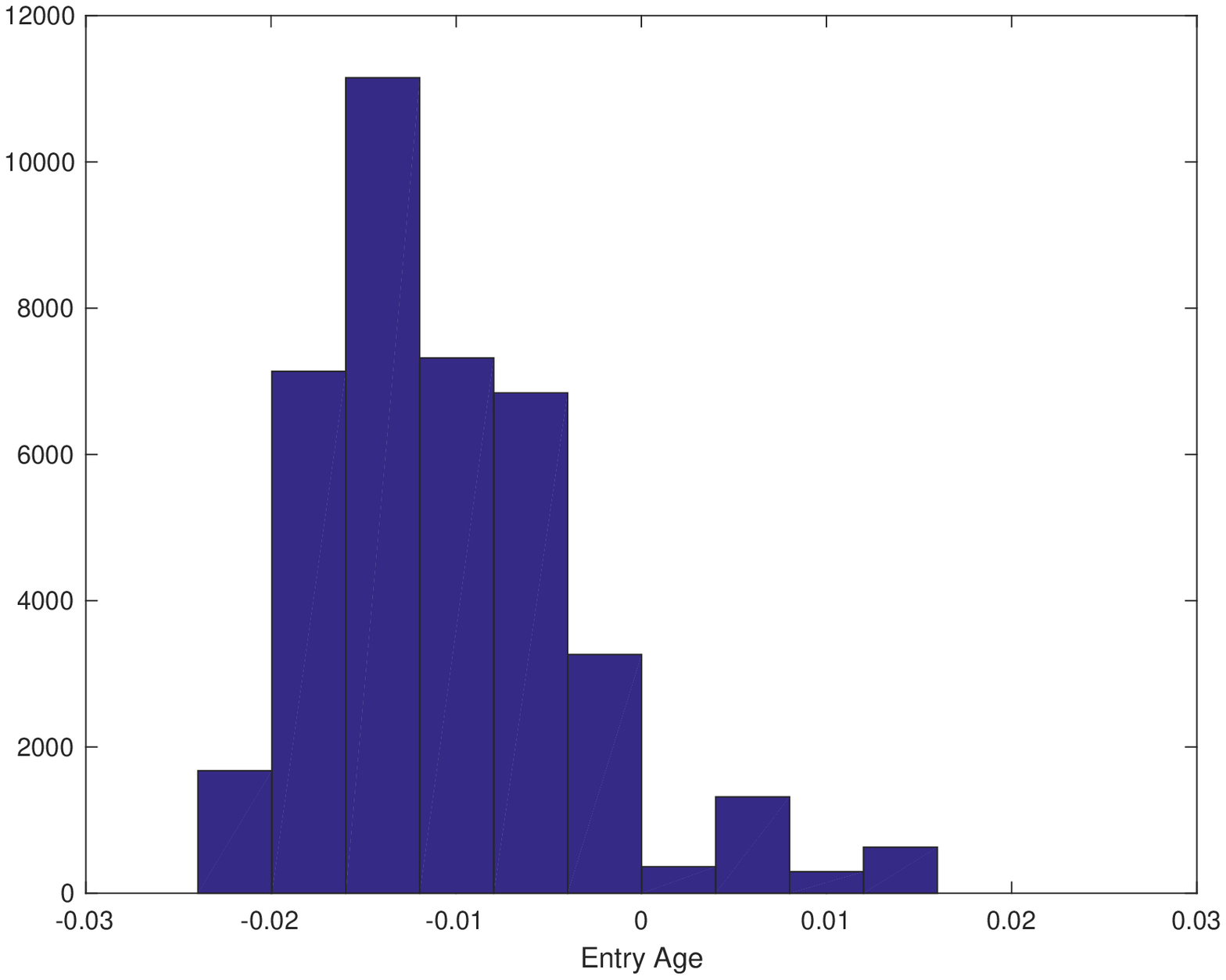} \\
\includegraphics[width=6.75cm]{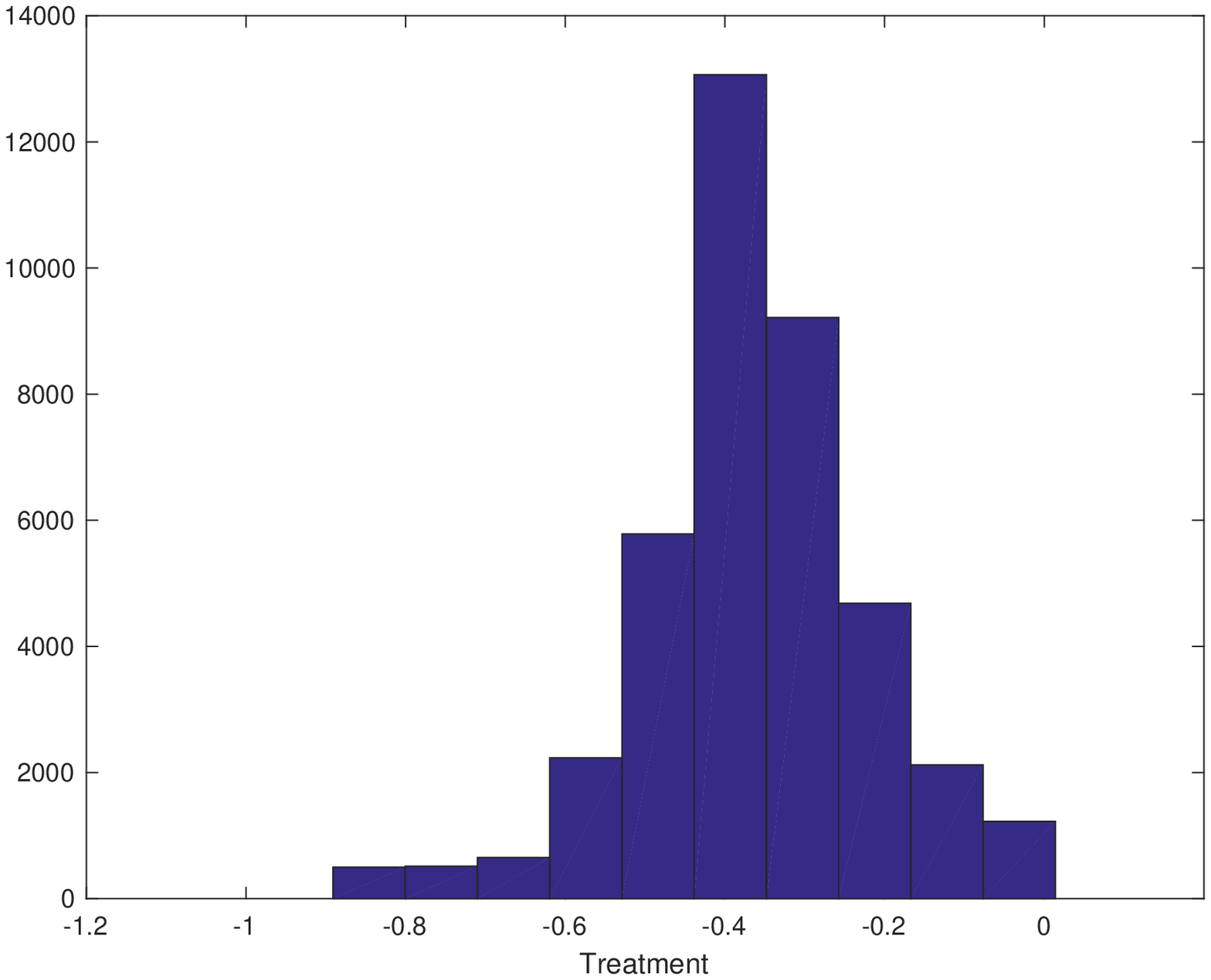} &
	\includegraphics[width=6.75cm]{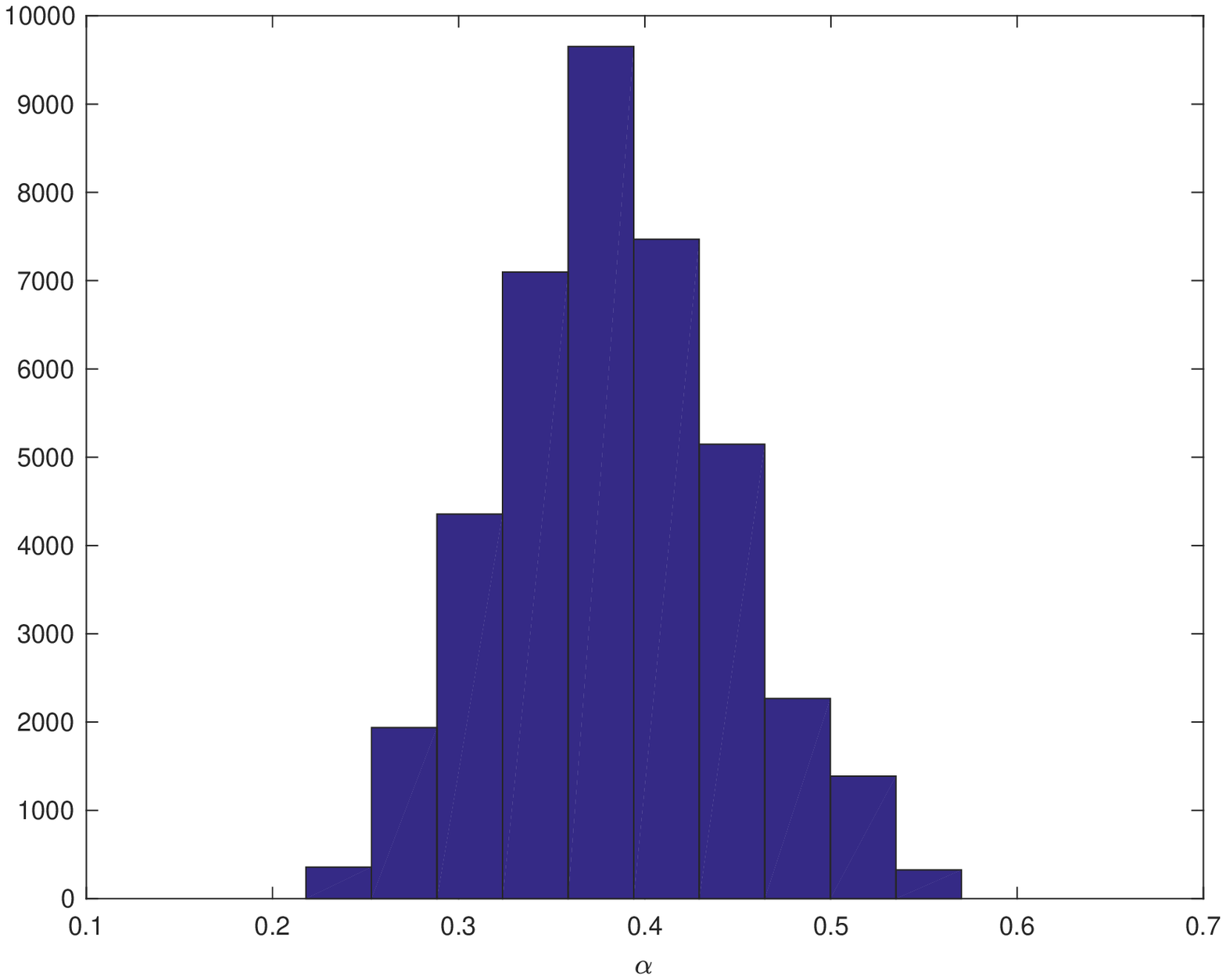} \\
\includegraphics[width=6.75cm]{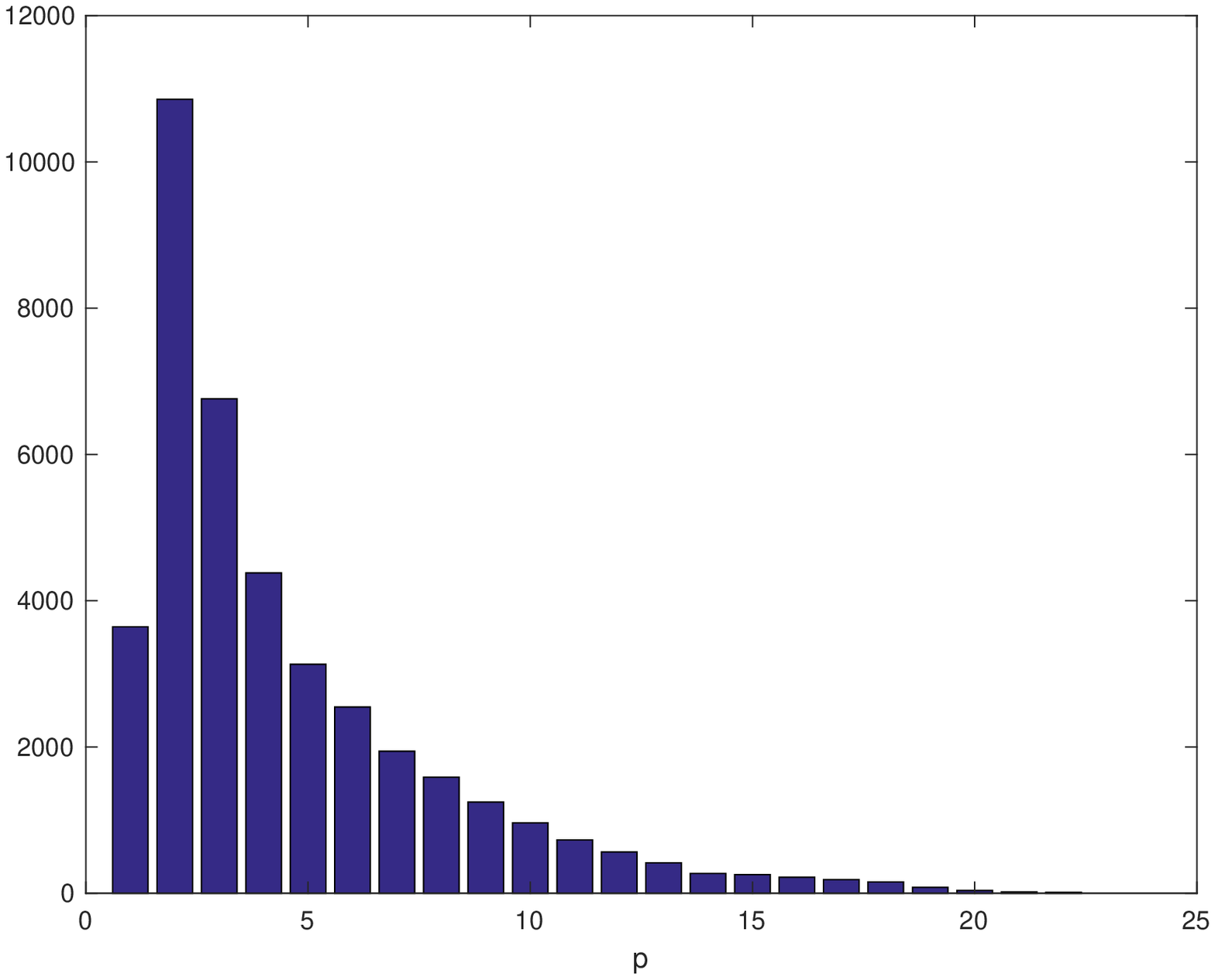} &
	\includegraphics[width=6.75cm]{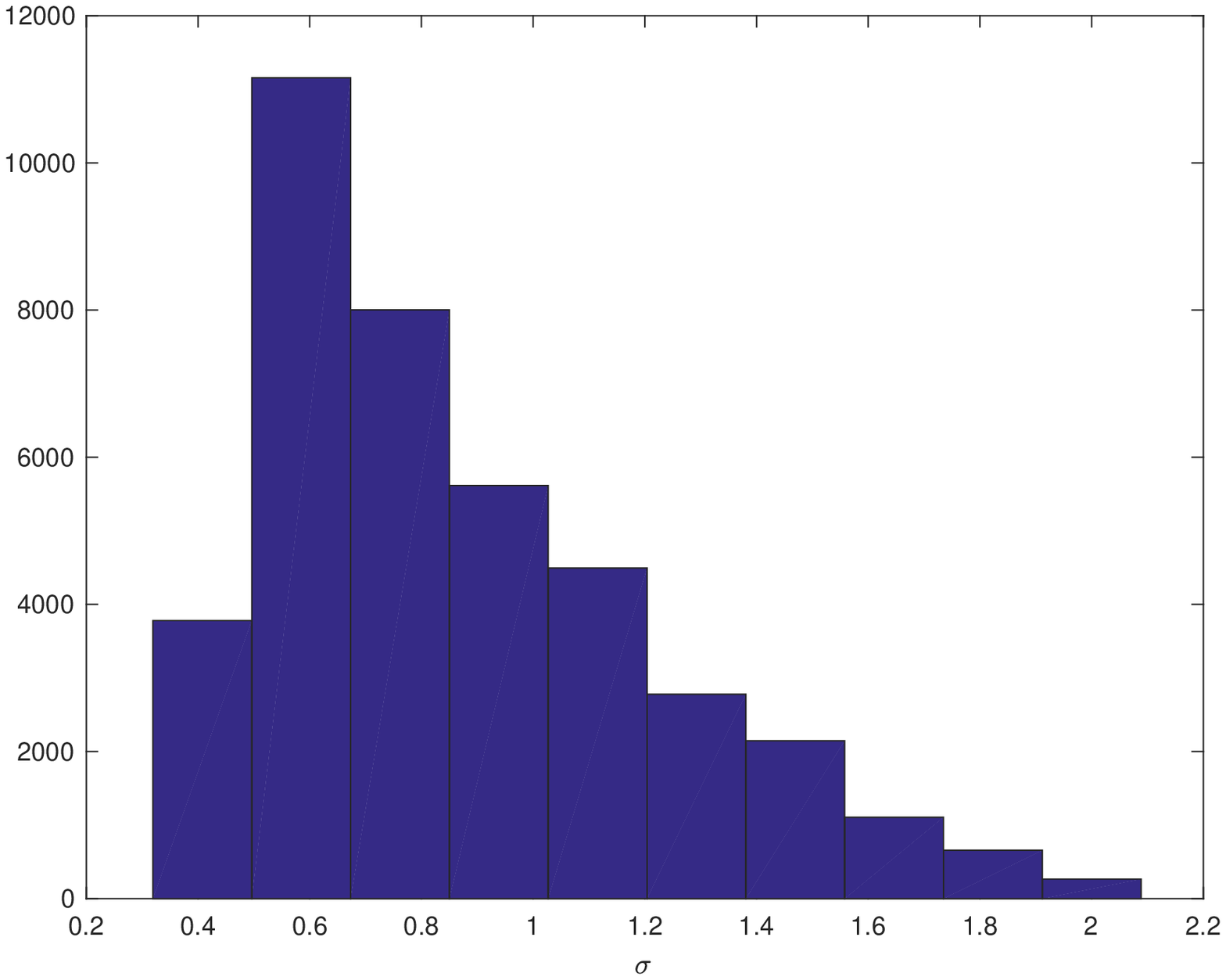} \\
\end{tabular}
\caption{Posterior histograms for the parameters of the regression model with SGLD errors for the SCLC study.}
\end{figure}\label{Fig:cancer}

The estimated intercept of the regression model, represented by the posterior median, is similar to the result in \cite{RubioYu2017}, that is 6.69. Similar considerations can be drawn for the coefficient of the entry age, which is very small and with a credible interval containing the zero; this last result supports the conclusion that the effect of the entry age on the survival time is negligible. However, the treatment appear to have a significant (negative) effect on the survival time, both under the estimated model and the results in \cite{RubioYu2017}. For the scale parameter $\sigma$, we again see agreement between the SEPD regression and the results of the above authors, although our credible interval is larger. It is not possible to perform a direct comparison of the estimated asymmetry, but in both case the value shows a clear positive skewness. Finally, our inferential procedure suggests that the data exhibit heavy tails, as indicated by the posterior median of $p=3$.

\begin{table}[h!]
\centering
\begin{tabular}{lccc|c}
\hline
Parameter & Mean & Median & $95\%$ C.I. & Rubio \& Yu\\
\hline
Intercept & 6.8438 & 6.9059 & (5.6692, 7.4672) & 6.690\\
Entry age & -0.0105 & -0.0117 & (-0.0216, 0.0079) & -0.009\\
Treatment & -0.3637 & -0.3611 & (-0.7161, -0.0552) & -0.446\\
$\alpha$ & 0.3837 & 0.3814 & (0.2717, 0.5116) & -0.395 ($\gamma$)\\
$p$ & 4.5349 & 3 & (1, 14) & NA\\
$\sigma$ & 0.8630 & 0.7771 & (0.3649, 1.7209) & 0.650 \\
\hline
\end{tabular}
\caption{SCLC Lung Cancer data: Posterior mean, posterior median and $95\%$ credible interval of the posterior for the regression model parameters. The last column to the right reports the posterior means from \cite{RubioYu2017}, where the skewness parameter is represented by $\gamma\in(-1,1)$ expressing positive skewness for values smaller than 0.}
\label{tab_SCLC}
\end{table}

\section{Discussion}
\label{Concl}
We have illustrated an objective Bayesian approach in the estimation of the tail parameter in two particular distributions: the skewed exponential power distribution (SEPD) and the skewed generalised logistic distribution (SGLD). This represents a new application of the well-known loss-based prior \citep{VillaWalker15}, where information theoretical considerations are used to derive minimally informative prior distributions. The SEPD and the SGLD are part of the wider family of two-piece location-scale distribution and allow to entangle skewness and tail fatness in one single probability distribution. Therefore, they represent an appealing modeling solution in scenarios where such behaviours are exhibited by the data, such as in financial applications and survival analysis. We illustrate the properties of the loss-based prior for the tail parameter of the above distributions by performing a thorough simulation study and analysis two real data sets. Furthermore, we show how the SEPD and SGLD can be used to model error terms in complex modeling situations, such as the error terms of autoregressive process for time series and error terms for linear regression models.

\section*{Acknowledgements}
Fabrizio Leisen was supported by the European Community's Seventh Framework Programme [FP7/2007-2013] under grant agreement no: 630677. \\
\noindent Luca Rossini acknowledges financial support from the European Union Horizon 2020 research and innovation programme under the Marie Sklodowska-Curie grant agreement no: 796902.
\vspace*{-8pt}

\bibliographystyle{apalike}
\bibliography{BNPAEDP}

\clearpage

\renewcommand{\thesection}{A}
\renewcommand{\theequation}{A.\arabic{equation}}
\renewcommand{\thefigure}{A.\arabic{figure}}
\renewcommand{\thetable}{A.\arabic{table}}
\setcounter{table}{0}
\setcounter{figure}{0}
\setcounter{equation}{0}

\section{Proofs}
\label{AppA}
\begin{proof}[Proof of Theorem \ref{ThKL}]
Note that,
\begin{equation}
D_{\text{KL}}\left(f_p^{\alpha,\mu,\sigma}||f_{p'}^{\alpha,\mu,\sigma}\right)=D_{\le}\left(f_p^{\alpha,\mu,\sigma}||f_{p'}^{\alpha,\mu,\sigma}\right)+D_{>}\left(f_p^{\alpha,\mu,\sigma}||f_{p'}^{\alpha,\mu,\sigma}\right),\notag
\end{equation}
where the above  Kullback--Leibler divergences are:
\begin{equation*}
D_{\le}\left(f_p^{\alpha,\mu,\sigma}||f_{p'}^{\alpha,\mu,\sigma}\right)= \int_{-\infty}^{\mu} f_p^{\alpha,\mu,\sigma}(y) \log{\left\{\frac{f_p^{\alpha,\mu,\sigma}(y)}{f_{p'}^{\alpha,\mu,\sigma}(y)}\right\}}  \,dy
\end{equation*}
and
$$ D_{\ge}\left(f_p^{\alpha,\mu,\sigma}||f_{p'}^{\alpha,\mu,\sigma}\right)= \int_{\mu}^{\infty} f_p^{\alpha,\mu,\sigma}(y) \log{\left\{\frac{f_p^{\alpha,\mu,\sigma}(y)}{f_{p'}^{\alpha,\mu,\sigma}(y)}\right\}}  \,dy.$$
Focusing on the first term, we have that:
\begin{equation*}
\begin{split}
D_{\le}\left(f_p^{\alpha,\mu,\sigma}||f_{p'}^{\alpha,\mu,\sigma}\right)&= \int_{-\infty}^{\mu} f_p^{\alpha,\mu,\sigma}(y) \log{\left\{\frac{f_p^{\alpha,\mu,\sigma}(y)}{f_{p'}^{\alpha,\mu,\sigma}(y)}\right\}}  \,dy\\
&= \int_{-\infty}^{\mu} \frac{1}{\sigma}f_p\left( \frac{y-\mu}{2\alpha\sigma}\right)\log{\left\{\frac{f_p\left( \frac{y-\mu}{2\alpha\sigma}\right)}{f_{p'}\left( \frac{y-\mu}{2\alpha\sigma}\right)}\right\}}  \,dy\\\
\end{split}
\end{equation*}
where the last equality follows from \eqref{GenTwoPiece}.  By performing the change of variable $x=\frac{y-\mu}{2\alpha\sigma}$, we get
\begin{equation*}
\begin{split}
D_{\le}\left(f_p^{\alpha,\mu,\sigma}||f_{p'}^{\alpha,\mu,\sigma}\right)&= 2\alpha \int_{-\infty}^0 f_{p}(x) \log{\left\{\frac{f_{p}(x)}{f_{p'}(x)}\right\}} \,dx \\
&= 2\alpha D_{\le}(f_{p}||f_{p'}). \\
\end{split}
\end{equation*}
For the models under consideration, i.e. the SEPD and the SGLD, $f_{p}=f_p^{\alpha=0.5,\mu=0,\sigma=1}$ and $f_{p'}=f_{p'}^{\alpha=0.5,\mu=0,\sigma=1}$. Therefore, 
$$D_{\le}\left(f_p^{\alpha,\mu,\sigma}||f_{p'}^{\alpha,\mu,\sigma}\right)=2\alpha D_{\le}\left(f_p^{\alpha=0.5,\mu=0,\sigma=1}||f_{p'}^{\alpha=0.5,\mu=0,\sigma=1}\right).$$
Similarly, 
$$D_{>}\left(f_p^{\alpha,\mu,\sigma}||f_{p'}^{\alpha,\mu,\sigma}\right)=2(1-\alpha)D_{>}\left(f_p^{\alpha=0.5,\mu=0,\sigma=1}||f_{p'}^{\alpha=0.5,\mu=0,\sigma=1}\right).$$
Note that,
$$D_{\le}\left(f_p^{\alpha=0.5,\mu=0,\sigma=1}||f_{p'}^{\alpha=0.5,\mu=0,\sigma=1}\right)=D_{>}\left(f_p^{\alpha=0.5,\mu=0,\sigma=1}||f_{p'}^{\alpha=0.5,\mu=0,\sigma=1}\right).$$ 
As a consequence,
\begin{equation*}
\begin{split}
D_{KL}\left(f_p^{\alpha=0.5,\mu=0,\sigma=1}||f_{p'}^{\alpha=0.5,\mu=0,\sigma=1}\right)&=2D_{\le}\left(f_p^{\alpha=0.5,\mu=0,\sigma=1}||f_{p'}^{\alpha=0.5,\mu=0,\sigma=1}\right)\\
&=2D_{>}\left(f_p^{\alpha=0.5,\mu=0,\sigma=1}||f_{p'}^{\alpha=0.5,\mu=0,\sigma=1}\right).
\end{split}
\end{equation*}
Therefore, by using the above identity, we can conclude the proof: 
\begin{align*}
D_{\text{KL}}\left(f_p^{\alpha,\mu,\sigma}||f_{p'}^{\alpha,\mu,\sigma}\right)&=D_{\le}\left(f_p^{\alpha,\mu,\sigma}||f_{p'}^{\alpha,\mu,\sigma}\right)+D_{>}\left(f_p^{\alpha,\mu,\sigma}||f_{p'}^{\alpha,\mu,\sigma}\right) \notag \\
&=2\alpha D_{\le}\left(f_p^{\alpha=0.5,\mu=0,\sigma=1}||f_{p'}^{\alpha=0.5,\mu=0,\sigma=1}\right) + \\
&\hspace{2cm}+2(1-\alpha)D_{>}\left(f_p^{\alpha=0.5,\mu=0,\sigma=1}||f_{p'}^{\alpha=0.5,\mu=0,\sigma=1}\right)\notag \\
&=\alpha D_{KL}\left(f_p^{\alpha=0.5,\mu=0,\sigma=1}||f_{p'}^{\alpha=0.5,\mu=0,\sigma=1}\right) + \\
&\hspace{2cm}+(1-\alpha)D_{KL}\left(f_p^{\alpha=0.5,\mu=0,\sigma=1}||f_{p'}^{\alpha=0.5,\mu=0,\sigma=1}\right)\notag \\
&=D_{KL}\left(f_p^{\alpha=0.5,\mu=0,\sigma=1}||f_{p'}^{\alpha=0.5,\mu=0,\sigma=1}\right). \notag
\end{align*}
\end{proof}
\begin{proof}[Proof of Theorem \ref{KLSEPD}]
Using the probability density function displayed in equation \eqref{SEPD}, we have that:
\begin{align}
D_{KL}\left(f_p^{\alpha=0.5,\mu=0,\sigma=1}||f_{p'}^{\alpha=0.5,\mu=0,\sigma=1}\right) &=\int_{-\infty}^{\infty} K(p) e^{-\frac{1}{p}\left|z\right|^{p}} \log{\left\{\frac{K(p)e^{-\frac{1}{p}\left|z\right|^{p}}}{K(p') e^{-\frac{1}{p'}\left|z\right|^{p'}}}\right\}}\,dz \notag \\
&=2\int_{0}^{\infty}K(p) e^{-\frac{1}{p}z^{p}} \log{\left\{\frac{K(p)e^{-\frac{1}{p}z^{p}}}{K(p') e^{-\frac{1}{p'}z^{p'}}}\right\}} \,dz \notag \\
&= 2\left[\log{K(p)}-\log{K(p')}\right]K(p) \int_{0}^{\infty}e^{-\frac{1}{p}z^{p}} \, dz \notag \\
&-2\frac{K(p)}{p} \int_{0}^{\infty}e^{-\frac{1}{p}z^{p}}z^{p} \,dz + 2\frac{K(p)}{p'} \int_{0}^{\infty}e^{-\frac{1}{p}z^{p}}z^{p'} \,dz. \label{pri}
\end{align}
From 3.478.1 of \cite{GR07},
\begin{equation}
\int_{0}^{\infty}x^{\nu-1}e^{-\mu x^{p}}\,dx=\frac{1}{p}\mu^{-\nu/p}\Gamma\left(\frac{\nu}{p}\right), \notag
\end{equation} 
we have that equation \ref{pri} becomes:
\begin{align}
D_{KL}\left(f_p^{\alpha=0.5,\mu=0,\sigma=1}||f_{p'}^{\alpha=0.5,\mu=0,\sigma=1}\right)&=2\left[\log{\frac{K(p)}{K(p')}}\right]\frac{K(p)}{p}\left(\frac{1}{p}\right)^{-\frac{1}{p}} \Gamma \left(\frac{1}{p}\right)-2\frac{K(p)}{p^2}\left(\frac{1}{p}\right)^{-\frac{(p+1)}{p}} \Gamma \left(\frac{p+1}{p}\right) \notag \\
&+2\frac{K(p)}{p'}\frac{1}{p}\left(\frac{1}{p}\right)^{-\frac{(p'+1)}{p}} \Gamma \left(\frac{p'+1}{p}\right) \notag \\
&=\log{K(p)}-\log{K(p')}-p^{-1}+\frac{p^{\frac{p'}{p}}}{p'}\frac{\Gamma \left(\frac{p'+1}{p}\right)}{\Gamma \left(\frac{1}{p}\right)}. \notag
\end{align}
In conclusion, following theorem \ref{ThKL}, the Kullback-Leibler divergence has the form
\begin{equation}
D_{\text{KL}}\left(f_p^{\alpha,\mu,\sigma}||f_{p'}^{\alpha,\mu,\sigma}\right)=\log{K(p)}-\log{K(p')}-p^{-1}+\frac{p^{\frac{p'}{p}}}{p'}\frac{\Gamma \left(\frac{p'+1}{p}\right)}{\Gamma \left(\frac{1}{p}\right)}.\notag
\end{equation}
\end{proof}
\begin{proof}[Proof of Theorem \ref{KLSLD}]
Using the probability density function displayed in equation \eqref{SLD}, we have that:
\begin{align}
D_{KL}\left(f_p^{\alpha=0.5,\mu=0,\sigma=1}||f_{p'}^{\alpha=0.5,\mu=0,\sigma=1}\right)&= \int_{-\infty}^{\infty} \frac{e^{-px}}{B(p,p)(1+e^{-x})^{2p}} \log{\left[\frac{B(p',p')}{B(p,p)}\frac{e^{-px}/(1+e^{-x})^{2p}}{e^{-p'x}/(1+e^{-x})^{2p'}}\right]} \,dx= \notag \\
&=\int_0^{1} \frac{1}{B(p,p)} t^{p-1}(1-t)^{p-1} \log{\left[\frac{B(p',p')}{B(p,p)} \frac{t^{p}(1-t)^{p}}{t^{p'}(1-t)^{p'}}\right]} \,dt, \notag
\end{align}
where $t=S(x)$, $dt=s(x)dx$ and $S(x)$ and $s(x)$ are the cumulative distribution and the probability distribution function of the logistic distribution, respectively. Hence, we can extend the Kullback--Leibler divergence as
\begin{align}
D_{KL}\left(f_p^{\alpha=0.5,\mu=0,\sigma=1}||f_{p'}^{\alpha=0.5,\mu=0,\sigma=1}\right)&=\log{\left[\frac{B(p',p')}{B(p,p)} \right]} + \frac{(p-p')}{B(p,p)} \int_0^1 t^{p-1}(1-t)^{p-1} \log{(t)} \,dt \notag \\
&+ \frac{(p-p')}{B(p,p)} \int_0^1 t^{p-1}(1-t)^{p-1} \log{(1-t)}\,dt \notag \\
&=\log{\left[\frac{B(p',p')}{B(p,p)} \right]} + \frac{2 (p-p')}{B(p,p)} \int_0^1 t^{p-1}(1-t)^{p-1} \log{(t)} \,dt. \label{DKL_Lo}
\end{align}
In conclusion, from 4.253.1 of \cite{GR07} 
\begin{equation}
\int_0^1 t^{p-1}(1-t)^{p-1} \log{(t)} dt=B(p,p)\left[\psi(p) -\psi(2p)\right], \notag
\end{equation} 
where $\psi(p)$ is the digamma function, the Kullback--Leibler divergence represented in equation \ref{DKL_Lo} have the following form
\begin{align}
D_{KL}\left(f_p^{\alpha=0.5,\mu=0,\sigma=1}||f_{p'}^{\alpha=0.5,\mu=0,\sigma=1}\right)&=\log{\left[\frac{B(p',p')}{B(p,p)} \right]} +2(p-p')\left[\psi(p)-\psi(2p)\right]. \notag 
\end{align}

\end{proof}

\clearpage

\clearpage

\renewcommand{\thesection}{B}
\renewcommand{\theequation}{B.\arabic{equation}}
\renewcommand{\thefigure}{B.\arabic{figure}}
\renewcommand{\thetable}{B.\arabic{table}}
\setcounter{table}{0}
\setcounter{figure}{0}
\setcounter{equation}{0}

\section{Comparison of the Kullback--Leibler Divergence for different distributions}

\begin{table}[h!]
\centering
\begin{tabular}{cccc|cccc}
\hline
p & $D_{\text{KL}}\left(f_p^{\alpha,\mu,\sigma}||f_{p-1}^{\alpha,\mu,\sigma}\right)$ & $D_{\text{KL}}\left(f_p^{\alpha,\mu,\sigma}||f_{p+1}^{\alpha,\mu,\sigma}\right)$ & & & p & $D_{\text{KL}}\left(f_p^{\alpha,\mu,\sigma}||f_{p-1}^{\alpha,\mu,\sigma}\right)$ & $D_{\text{KL}}\left(f_p^{\alpha,\mu,\sigma}||f_{p+1}^{\alpha,\mu,\sigma}\right)$ \\
\hline
2 & 7.2093$\times 10^{-02}$ & 5.9144$\times 10^{-02}$ & & & 17& 6.1438 $\times 10^{-04}$ & 6.3413$\times 10^{-04}$ \\
3 & 2.7675$\times 10^{-02}$ & 2.6462$\times 10^{-02}$ & & & 18& 5.3875 $\times 10^{-04}$ & 5.5578 $\times 10^{-04}$ \\
4 & 1.4752$\times 10^{-02}$ & 1.4759$\times 10^{-02}$ & & &19&  4.7558 $\times 10^{-04}$ & 4.9036 $\times 10^{-04}$ \\
5 & 9.1321$\times 10^{-03}$ & 9.3019$\times 10^{-03}$ & & &20& 4.2233 $\times 10^{-04}$& 4.3523 $\times 10^{-04}$\\
6 & 6.1732$\times 10^{-03}$ & 6.3402$\times 10^{-03}$ & & &21& 3.7707 $\times 10^{-04}$& 3.8840 $\times 10^{-04}$ \\
7 & 4.4262$\times 10^{-03}$ & 4.5646$\times 10^{-03}$ & & &22& 3.3833 $\times 10^{-04}$& 3.4833 $\times 10^{-04}$ \\
8 & 3.3117$\times 10^{-03}$ &3.4223$\times 10^{-03}$ & & &23& 3.0494 $\times 10^{-04}$& 3.1381 $\times 10^{-04}$ \\
9 & 2.5594$\times 10^{-03}$ & 2.6474$\times 10^{-03}$ & & &24& 2.7599 $\times 10^{-04}$& 2.8388$\times 10^{-04}$ \\
10& 2.0292$\times 10^{-03}$ & 2.0997$\times 10^{-03}$ & & &25& 2.5074 $\times 10^{-04}$ & 2.5780 $\times 10^{-04}$ \\
11& 1.6426$\times 10^{-03}$ & 1.6996$\times 10^{-03}$ & & &26& 2.2861 $\times 10^{-04}$& 2.3494 $\times 10^{-04}$ \\
12& 1.3527$\times 10^{-03}$ & 1.3994$\times 10^{-03}$ & & &27& 2.0911 $\times 10^{-04}$& 2.1481 $\times 10^{-04}$ \\
13& 1.1302$\times 10^{-03}$ & 1.1688$\times 10^{-03}$ & & &28& 1.9186 $\times 10^{-04}$ & 1.9701 $\times 10^{-04}$ \\
14& 9.5616$\times 10^{-04}$ & 9.8838 $\times 10^{-04}$ & & &29& 1.7653 $\times 10^{-04}$ & 1.8120 $\times 10^{-04}$ \\
15& 8.1765 $\times 10^{-04}$ & 8.4480 $\times 10^{-04}$ & & &30& 1.6285 $\times 10^{-04}$& 1.6710 $\times 10^{-04}$ \\
16& 7.0582 $\times 10^{-04}$ & 7.2889 $\times 10^{-04}$ & & & & & \\
&  &  & & &  & &  \\
\hline
&  &  & & &  & &  \\
30 & 1.6285 $\times 10^{-04}$& 1.6710 $\times 10^{-04}$  & & & 120 & 5.3155 $\times 10^{-06}$ & 5.3800 $\times 10^{-06}$ \\
 60 & 3.0272 $\times 10^{-05}$ & 3.0829 $\times 10^{-05}$ & & & 150 & 3.0055 $\times 10^{-06}$ & 3.0370 $\times 10^{-06}$\\
90 & 1.1010 $\times 10^{-05}$ & 1.1169 $\times 10^{-05}$  & & & 180 & 1.8801 $\times 10^{-06}$ & 1.8975 $\times 10^{-06}$ \\
\hline
\end{tabular}
\caption{Comparison of the Kullbck-Leibler divergences $D_{\text{KL}}\left(f_p^{\alpha,\mu,\sigma}||f_{p-1}^{\alpha,\mu,\sigma}\right)$ and $D_{\text{KL}}\left(f_p^{\alpha,\mu,\sigma}||f_{p+1}^{\alpha,\mu,\sigma}\right)$ for $p=2,\dots,30$ for the SEPD.}
\label{TKL30}
\end{table}

%

\begin{table}[h!]
\centering
\begin{tabular}{cccc|cccc}
\hline
p & $D_{\text{KL}}\left(f_p^{\alpha,\mu,\sigma}||f_{p-1}^{\alpha,\mu,\sigma}\right)$ & $D_{\text{KL}}\left(f_p^{\alpha,\mu,\sigma}||f_{p+1}^{\alpha,\mu,\sigma}\right)$ & & & p & $D_{\text{KL}}\left(f_p^{\alpha,\mu,\sigma}||f_{p-1}^{\alpha,\mu,\sigma}\right)$ & $D_{\text{KL}}\left(f_p^{\alpha,\mu,\sigma}||f_{p+1}^{\alpha,\mu,\sigma}\right)$ \\
\hline
2 & 0.1251 & 0.0572 & & & 17& 9.2755 $\times 10^{-04}$ & 8.5657 $\times 10^{-04}$ \\
3 & 0.0428 & 0.0262 & & &18& 8.2411 $\times 10^{-04}$ & 7.6446 $\times 10^{-04}$ \\
4 & 0.0214 & 0.0150 & & & 19& 7.3704 $\times 10^{-04}$ & 6.8644 $\times 10^{-04}$ \\
5 & 0.0128 & 0.0097 & & & 20& 6.6308 $\times 10^{-04}$ & 6.1979 $\times 10^{-04}$ \\
6 & 0.0085 & 0.0068 & & & 21& 5.9972 $\times 10^{-04}$ & 5.6239 $\times 10^{-04}$ \\
7 & 0.0061 & 0.0050 & & & 22& 5.4503 $\times 10^{-04}$ & 5.1261 $\times 10^{-04}$ \\
8 & 0.0045 & 0.0038 & & & 23& 4.9749 $\times 10^{-04}$ & 4.6916 $\times 10^{-04}$ \\
9 & 0.0035 & 0.0030 & & & 24& 4.5591 $\times 10^{-04}$ & 4.3101 $\times 10^{-04}$ \\
10& 0.0028 & 0.0025 & & &25& 4.1933 $\times 10^{-04}$ & 3.9733 $\times 10^{-04}$ \\
11& 0.0023 & 0.0020 & & & 26& 3.8698 $\times 10^{-04}$ & 3.6745 $\times 10^{-04}$ \\
12& 0.0019 & 0.0017 & & & 27& 3.5824 $\times 10^{-04}$ & 3.4082 $\times 10^{-04}$ \\
13& 0.0016 & 0.0015 & & & 28& 3.3258 $\times 10^{-04}$ & 3.1698 $\times 10^{-04}$ \\
14& 0.0014 & 0.0013 & & & 29& 3.0959 $\times 10^{-04}$ & 2.9556 $\times 10^{-04}$ \\
15& 0.0012 & 0.0011 & & & 30& 2.8890 $\times 10^{-04}$ & 2.7623 $\times 10^{-04}$ \\
16& 0.0011 & 0.0010 & & &  & &  \\
&  &  & & &  & &  \\
\hline
&  &  & & &  & &  \\
30 & 2.8890 $\times 10^{-04}$ & 2.7623 $\times 10^{-04}$  & & & 120 & 1.7531 $\times 10^{-05}$ & 1.7337 $\times 10^{-05}$  \\
60 & 7.0814 $\times 10^{-05}$ & 6.9252 $\times 10^{-05}$ & & & 150 & 1.1198 $\times 10^{-05}$ & 1.1099 $\times 10^{-05}$ \\
90 & 3.1268 $\times 10^{-05}$ & 3.0807 $\times 10^{-05}$  & & & 180 & 7.7663 $\times 10^{-06}$ & 7.7089 $\times 10^{-06}$  \\
\hline
\end{tabular}
\caption{Comparison of the Kullbck-Leibler divergences $D_{\text{KL}}\left(f_p^{\alpha,\mu,\sigma}||f_{p-1}^{\alpha,\mu,\sigma}\right)$ and $D_{\text{KL}}\left(f_p^{\alpha,\mu,\sigma}||f_{p+1}^{\alpha,\mu,\sigma}\right)$ for $p=2,\dots,30$ for the SGLD.}
\label{TKL_LD30}
\end{table}

\clearpage

\renewcommand{\thesection}{C}
\renewcommand{\theequation}{C.\arabic{equation}}
\renewcommand{\thefigure}{C.\arabic{figure}}
\renewcommand{\thetable}{C.\arabic{table}}
\setcounter{table}{0}
\setcounter{figure}{0}
\setcounter{equation}{0}

\section{Generate samples with the SEPD and the SGLD}
\label{App_sampling}

As described by \cite{Kom07} and \cite{ZhuZinde09}, the SEPD could  be simulated by using the standard skewed exponential power distribution. They highlighted the following procedure: 
\begin{itemize}
\item Generate $U\sim U (0,1)$, $W_1\sim Ga (1/p,1)$ and $W_2\sim Ga (1/p,1)$.
\item Compute the following transformed variable, 
\begin{equation}
Z=\alpha W_1^{1/p}\left[\frac{\text{sign}(U-\alpha)-1}{2K(p)\Gamma(1+1/p)}\right] + (1-\alpha) W_2^{1/p}\left[\frac{\text{sign}(U-\alpha)+1}{2K(p)\Gamma(1+1/p)}\right].\label{Sim_SEPD}
\end{equation}
where $\text{sign}(x)=1$ if $x>0$, $\text{sign}(x)=-1$ if $x<0$ and $\text{sign}(x)=0$ if $x=0$. 
 \end{itemize}
The above procedure generates an observation from the SEPD with $\mu=0$, $\sigma=1$. We can obtain a SEPD with different $\mu$ and $\sigma$ by simply multiplying $Z$ by the scale parameter and adding the location parameter. 

\medskip

\noindent Simulating from the SGLD follows a similar procedure. In particular, 
\begin{itemize}
\item Generate $U\sim U (0,1)$,  $W_1$ and $W_2$ from the distribution displayed in \eqref{BB}.
\item Compute the following transformed variable,
\begin{equation}
Z=\alpha |W_1|\left[\text{sign}(U-\alpha)-1\right] + (1-\alpha) |W_2|\left[\text{sign}(U-\alpha)+1\right],\label{Sim_SGLD}
\end{equation}
where $\text{sign}(x)=1$ if $x>0$, $\text{sign}(x)=-1$ if $x<0$ and $\text{sign}(x)=0$ if $x=0$. 
\end{itemize}
 The above procedure generates an observation from the SGLD with $\mu=0$, $\sigma=1$. We can obtain a SGLD with different $\mu$ and $\sigma$ by simply multiplying $Z$ by the scale parameter and adding the location parameter.

\end{document}